\documentclass[11pt]{article}

\pdfoutput=1
\usepackage[protrusion=true,expansion=true]{microtype}
\usepackage{amsmath,amssymb,amsfonts,amsthm}
\usepackage{subcaption}
\usepackage{graphicx}
\usepackage[backref=page]{hyperref}
\usepackage{color}
\usepackage{wrapfig}
\usepackage{tikz}
\usetikzlibrary{decorations.pathreplacing}
\usepackage{setspace}
\usepackage{algorithm}
\usepackage[noend]{algpseudocode}
\usepackage[framemethod=tikz]{mdframed}
\usepackage{xspace}
\usepackage{pgfplots}
\usepackage{framed}
\usepackage{thmtools}
\usepackage{thm-restate}
\pgfplotsset{compat=1.5}

\newtheorem{theorem}{Theorem}[section]
\newtheorem{corollary}[theorem]{Corollary}
\newtheorem{lemma}[theorem]{Lemma}

\newtheorem{definition}[theorem]{Definition}

\newtheorem{claim}[theorem]{Claim}

\newtheorem{fact}[theorem]{Fact}

\newenvironment{proofof}[1]{\begin{trivlist} \item {\bf Proof
#1:~~}}
  {\qed\end{trivlist}}

\newcommand{\namedref}[2]{\hyperref[#2]{#1~\ref*{#2}}}
\newcommand{\thmlab}[1]{\label{thm:#1}}
\newcommand{\thmref}[1]{\namedref{Theorem}{thm:#1}}
\newcommand{\lemlab}[1]{\label{lem:#1}}
\newcommand{\lemref}[1]{\namedref{Lemma}{lem:#1}}
\newcommand{\claimlab}[1]{\label{claim:#1}}
\newcommand{\claimref}[1]{\namedref{Claim}{claim:#1}}
\newcommand{\corlab}[1]{\label{cor:#1}}
\newcommand{\corref}[1]{\namedref{Corollary}{cor:#1}}
\newcommand{\seclab}[1]{\label{sec:#1}}
\newcommand{\secref}[1]{\namedref{Section}{sec:#1}}
\newcommand{\applab}[1]{\label{app:#1}}
\newcommand{\appref}[1]{\namedref{Appendix}{app:#1}}
\newcommand{\factlab}[1]{\label{fact:#1}}
\newcommand{\factref}[1]{\namedref{Fact}{fact:#1}}

\newcommand{\figlab}[1]{\label{fig:#1}}
\newcommand{\figref}[1]{\namedref{Figure}{fig:#1}}
\newcommand{\alglab}[1]{\label{alg:#1}}
\renewcommand{\algref}[1]{\namedref{Algorithm}{alg:#1}}

\newcommand{\deflab}[1]{\label{def:#1}}
\newcommand{\defref}[1]{\namedref{Definition}{def:#1}}

\def \Lap    {\mdef{\mathsf{Lap}}}

\newcommand{\PPr}[1]{\ensuremath{\mathbf{Pr}\left[#1\right]}}

\renewcommand{\O}[1]{\ensuremath{\mathcal{O}\left(#1\right)}}

\newcommand{\eps}{\varepsilon}

\def \ams    {\mdef{\textsc{AMS}}}

\def \countmin    {\mdef{\textsc{CountMin}}}
\def \misragries    {\mdef{\textsc{MisraGries}}}
\def \countsketch    {\mdef{\textsc{CountSketch}}}
\def \counter    {\mdef{\textsc{Counter}}}
\def \Lap    {\mdef{\mathsf{Lap}}}

\def \calA    {\mdef{\mathcal{A}}}

\def \calE    {\mdef{\mathcal{E}}}

\def \calL    {\mdef{\mathcal{L}}}
\def \calM    {\mdef{\mathcal{M}}}
\def \calS    {\mdef{\mathcal{S}}}
\def \calT    {\mdef{\mathcal{T}}}
\def \calU    {\mdef{\mathcal{U}}}

\def \frakS    {\mdef{\mathfrak{S}}}

\newcommand{\mdef}[1]{{\ensuremath{#1}}\xspace}  

\DeclareMathOperator*{\poly}{poly}

\newcommand{\ceil}[1]{\mdef{\left\lceil#1\right\rceil}}               

\newcommand{\ignore}[1]{}

\newif\ifnotes\notestrue 
\ifnotes
\newcommand{\samson}[1]{\textcolor{green}{{\bf (Samson:} {#1}{\bf ) }} \marginpar{\tiny\bf
             \begin{minipage}[t]{0.5in}
               \raggedright S:
            \end{minipage}}}   
\newcommand{\jeremiah}[1]{\textcolor{purple}{{\bf (Jeremiah:} {#1}{\bf ) }} \marginpar{\tiny\bf
             \begin{minipage}[t]{0.5in}
               \raggedright J:
            \end{minipage}}}           
\newcommand{\tanote}[1]{\textcolor{purple}{{\bf (Tamalika:} {#1}{\bf ) }} \marginpar{\tiny\bf
             \begin{minipage}[t]{0.5in}
               \raggedright T:
            \end{minipage}}}
\newcommand{\seunghoon}[1]{\textcolor{red}{{\bf (Seunghoon:} {#1}{\bf ) }} \marginpar{\tiny\bf
             \begin{minipage}[t]{0.5in}
               \raggedright L:
            \end{minipage}}}

\else
\newcommand{\samson}[1]{}
\newcommand{\jeremiah}[1]{}
\newcommand{\tanote}[1]{}
\newcommand{\seunghoon}[1]{}

\fi

\providecommand{\email}[1]{\href{mailto:#1}{\nolinkurl{#1}\xspace}}

\definecolor{darkpastelgreen}{rgb}{0.01, 0.75, 0.24}
\definecolor{bleudefrance}{rgb}{0.19, 0.55, 0.91}
\hypersetup{
     colorlinks   = true,
     citecolor    = bleudefrance,
	 linkcolor	  = darkpastelgreen,
	 urlcolor     = darkpastelgreen
}

\makeatletter
\renewcommand*{\@fnsymbol}[1]{\textcolor{darkpastelgreen}{\ensuremath{\ifcase#1\or *\or \dagger\or \ddagger\or
 \mathsection\or \triangledown\or \mathparagraph\or \|\or **\or \dagger\dagger
   \or \ddagger\ddagger \else\@ctrerr\fi}}}
\makeatother
\usepackage{fullpage}
\begin{document}

\title{Differentially Private $L_2$-Heavy Hitters in the Sliding Window Model}


\author{Jeremiah Blocki\thanks{Purdue University. Supported in part by NSF CCF-1910659, NSF CNS-1931443, and NSF CAREER award CNS-2047272. 
E-mail: \email{jblocki@purdue.edu}}
\and
Seunghoon Lee\thanks{Purdue University. Supported by NSF CAREER award CNS-2047272.  
E-mail: \email{lee2856@purdue.edu}}
\and
Tamalika Mukherjee\thanks{Purdue University. Supported in part by Purdue Bilsland Dissertation Fellowship, NSF CCF-1910659, and NSF CCF-2228814. 
E-mail: \email{tmukherj@purdue.edu}}
\and
Samson Zhou\thanks{UC Berkeley and Rice University. Work done in part while at Carnegie Mellon University and supported by a Simons Investigator Award of David P. Woodruff and by the National Science Foundation under Grant No. CCF-1815840. 
E-mail: \email{samsonzhou@gmail.com}}
}
\date{\today}


\date{\today}
\maketitle

\begin{abstract}
The data management of large companies often prioritize more recent data, as a source of higher accuracy prediction than outdated data. For example, the Facebook data policy retains user search histories for $6$ months while the Google data retention policy states that browser information may be stored for up to $9$ months. These policies are captured by the sliding window model, in which only the most recent $W$ statistics form the underlying dataset. 

In this paper, we consider the problem of privately releasing the $L_2$-heavy hitters in the sliding window model, which include $L_p$-heavy hitters for $p\le 2$ and in some sense are the strongest possible guarantees that can be achieved using polylogarithmic space, but cannot be handled by existing techniques due to the sub-additivity of the $L_2$ norm. Moreover, existing non-private sliding window algorithms use the smooth histogram framework, which has high sensitivity. 

To overcome these barriers, we introduce the first differentially private algorithm for $L_2$-heavy hitters in the sliding window model by initiating a number of $L_2$-heavy hitter algorithms across the stream with significantly lower threshold. Similarly, we augment the algorithms with an approximate frequency tracking algorithm with significantly higher accuracy. We then use smooth sensitivity and statistical distance arguments to show that we can add noise proportional to an estimation of the $L_2$ norm. To the best of our knowledge, our techniques are the first to privately release statistics that are related to a sub-additive function in the sliding window model, and may be of independent interest to future differentially private algorithmic design in the sliding window model.
\end{abstract}

\section{Introduction}
Differential privacy~\cite{Dwork06,DworkMNS16} has emerged as the standard for privacy in the both the research and industrial communities. 
For example, Google Chrome uses RAPPOR~\cite{ErlingssonPK14} to collect user statistics such as the default homepage of the browser or the default search engine, etc., Samsung proposed a similar mechanism to collect numerical answers such as the time of usage and battery volume~\cite{NguyenXYHSS16}, and Apple uses a differentially private method~\cite{greenberg2016apple} to generate predictions of spellings. 

The age of collected data can significantly impact its relevance to predicting future patterns, as the behavior of groups or individuals may significantly change over time due to either cyclical, temporary, or permanent change. 
Indeed, recent data is often a more accurate predictor than older data across multiple sources of big data, such as stock markets or Census data, a concept which is often reflected through the data management of large companies. 
For example, the Facebook data policy~\cite{FB-data} retains user search histories for $6$ months, the Apple differential privacy~\cite{Upadhyay19} states that collected data is retained for $3$ months, the Google data retention policy states that browser information may be stored for up to $9$ months~\cite{google-data}, and more generally, large data collection agencies often perform analysis and release statistics on time-bounded data. 
However, since large data collection agencies often manage highly sensitive data, the statistics must be released in a way that does not compromise privacy. 
Thus in this paper, we study the (event-level) differentially private release of statistics of time-bounded data that only use space sublinear in the size of the data. 

\begin{definition}[Differential privacy~\cite{DworkMNS16}]
\deflab{def:dp}
Given $\eps>0$ and $\delta\in(0,1)$, a randomized algorithm $\calA$ operating on datastreams is \emph{$(\eps,\delta)$-differentially private} if, for every pair of neighboring datasets $\frakS$ and $\frakS'$ and for all sets $E$ of possible outputs, we have,
\[\PPr{\calA(\frakS)\in E}\le e^{\eps}\cdot\PPr{\calA(\frakS')\in E}+\delta.\]
\end{definition}

In the popular streaming model of computation, elements of an underlying dataset arrive one-by-one but the entire dataset is considered too large to store; thus algorithms are restricted to using space sublinear in the size of the data. 
Although the streaming model provides a theoretical means to handle big data and has been studied thoroughly for applications in privacy-preserving data analysis, e.g.,~\cite{MirMNW11,BlockiBDS12,JosephRUW20,HuangQYC21,DinurSWZ23} and adaptive data analysis, e.g.,~\cite{AvdiukhinMYZ19,Ben-EliezerJWY22,HassidimKMMS20,BravermanHMSSZ21,ChakrabartiGS22,AjtaiBJSSWZ22,BeimelKMNSS22,Ben-EliezerEO22,AttiasCSS23}, it does not properly capture the ability to prioritize more recent data, which is a desirable quality for data summarization. 
The time decay model~\cite{CohenS06,KopelowitzP08,SuYC18,BravermanLUZ19} emphasizes more recent data by assigning a polynomially decaying or exponentially decaying weight to ``older'' data points, but these functions cannot capture the zero-one property when data older than a certain age is completely deleted. 

\paragraph{The sliding window model.} 
By contrast, the \emph{sliding window model} takes a large data stream as an input and only focuses on the updates past a certain point in time by implicitly defining the underlying dataset through the most recent $W$ updates of the stream, where $W>0$ is the window parameter. 
Specifically, given a stream $u_1,\ldots,u_m$ such that $u_i\in[n]$ for all $i\in[m]$ and a parameter $W>0$ that we assume satisfies $W\le m$ without loss of generality, the underlying dataset is a frequency vector $f\in\mathbb{R}^n$ induced by the last $W$ updates of the stream $u_{m-W+1},\ldots,u_m$ so that 
\[f_k=|\{i\,:u_i=k\}|,\]
for all $k\in[n]$. 
Then the goal is to output a private approximation to the frequency $f_k$ of each heavy-hitter, i.e., the indices $k\in[n]$ for which $f_k\ge\alpha L_p(f)$, which denotes the $L_p$ norm of $f$ for a parameter $p\ge 1$: 
\[L_p(f)=\|f\|_p=\left(\sum_{i=1}^n f_i^p\right)^{1/p}.\]
In this case, we say that streams $\frakS$ and $\frakS'$ are neighboring if there exists a single update $i\in[m]$ such that $u_i\neq u'_i$, where $u_1,\ldots,u_m$ are the updates of $\frakS$ and $u'_1,\ldots,u'_m$ are the updates of $\frakS'$.

Note that if $k$ is an $L_1$-heavy hitter, i.e., a heavy-hitter with respect to $L_1(f)$, then $f_k\ge\alpha L_1(f)$ so that
\[f_k\ge\alpha\left(\sum_{i=1}^n f_i\right)\ge\alpha\left(\sum_{i=1}^n f_i^2\right)^{1/2},\]
and $k$ is also an $L_2$-heavy hitter. 
Thus, any $L_2$-heavy hitter algorithm will also report the $L_1$-heavy hitters, but the converse is not always true.    
Indeed, for the Yahoo! password frequency corpus~\cite{BlockiDB16} ($n\approx 70$ million) with heavy-hitter threshold $\alpha=\frac{1}{500}$ there were $3,972$ $L_2$-heavy hitters, but only one $L_1$-heavy hitter. 
On the other hand, finding $L_p$-heavy hitters for $p>2$ requires $\Omega(n^{1-2/p})$ space~\cite{ChakrabartiKS03,Bar-YossefJKS04}, so in some sense, the $L_2$-heavy hitters are the best we can hope to find using polylogarithmic space. 
Although there is a large and active line of work in the sliding window model~\cite{DatarGIM02,BravermanO07,BravermanGO14,BravermanLLM16,BravermanGLWZ18,BravermanDMMUWZ20,BorassiELVZ20,WoodruffZ21,BravermanWZ21,JayaramWZ22}, there is surprisingly little work in the sliding window model that considers differential privacy~\cite{Upadhyay19,UpadhyayU21}.

\subsection{Our Contributions}
In this paper, we consider the problem of privately releasing approximate frequencies for the heavy-hitters in a dataset defined by the sliding window model. 
We give the first differentially private algorithm for approximating the frequencies of the $L_2$-heavy hitters in the sliding window model.
\begin{restatable}{theorem}{thmswmain}
\thmlab{thm:sw:main}
For any $\alpha\in(0,1), c>0$, window parameter $W$ on a stream of length $m$ that induces a frequency vector $f\in\mathbb{R}^n$ in the sliding window model, and privacy parameter $\eps>\frac{1000\log m}{\alpha^3\sqrt{W}}$, there exists an algorithm such that:
\begin{enumerate}
\item (Privacy)
The algorithm is $(\eps,\delta)$-differentially private for $\delta=\frac{1}{m^c}$. 
\item (Heavy-hitters)
With probability at least $1-\frac{1}{m^c}$, the algorithm outputs a list $\calL$ such that $k\in\calL$ for each $k\in[n]$ with $f_k\ge\alpha\,L_2(f)$ and $j\notin\calL$ for each $j\in[n]$ with $f_j\le\frac{\alpha}{2}\,L_2(f)$. 
\item (Accuracy)
With probability at least $1-\frac{1}{m^c}$, we simultaneously have $|f_k-\widetilde{f_k}|\le\frac{\alpha}{4}\,L_2(f)$ for all $k\in\calL$, where $\widetilde{f_k}$ denotes the noisy approximation of $f_k$ output by the algorithm.
\item (Complexity)
The algorithm uses $\O{\frac{\log^7 m}{\alpha^6\eta^4}}$ bits of space and $\O{\frac{\log^4 m}{\alpha^3\eta^4}}$ operations per update where $\eta = \max\{1, \eps\}$. 
\end{enumerate}
\end{restatable}

Along the way, we develop techniques for handling differentially private heavy-hitter algorithms in the sliding window model that may be of independent interest. 
In particular, we also use our techniques to obtain an $L_1$-heavy hitter algorithm for the sliding window model that guarantees \emph{pure} differential privacy. 
Finally, we give an algorithm for continual release of $L_1$ and $L_2$-heavy hitters in the sliding window model that has additive error $\frac{\alpha\sqrt{W}}{2}$ for each estimated heavy-hitter frequency and preserves pure differential privacy, building on a line of work~\cite{ChanLSX12,Upadhyay19,HuangQYC21} for continual release. 
By comparison, the algorithm of \cite{Upadhyay19} guarantees $\O{W^{3/4}}$ additive error while the algorithm of \cite{HuangQYC21} gives $(\eps,\delta)$-differential privacy.  
We remark that since $\sqrt{W}\le L_2(t-W+1:t)$ for any $t\in[m]$, where $L_2(t-W+1:t)$ denotes the $L_2$ norm of the sliding window between times $t-W+1$ and $t$, then our improvements over \cite{Upadhyay19} for the continual release of $L_1$-heavy hitters actually also resolve the problem of continual release of $L_2$-heavy hitters.  
Nevertheless, the approach is somewhat standard and thus we defer discussion to the appendix.

\subsection{Related Work}
\paragraph{Dynamic structures vs. linear sketching. }
Non-private algorithms in the streaming model generally follow one of two main approaches. 
The first main approach is the transformation from static data structures to dynamic structures using the framework of~\cite{BentleyS80}. 
Although the approach has been a useful tool for many applications~\cite{DworkNPRY10,ChanSS11,ChanLSX12,LarsenMWY20}, it does provide a mechanism to handle the implicit deletion of updates induced by the sliding window model. 
The second main approach is the use of linear sketching~\cite{BlockiBDS12,BassilyS15,BunNS19,BassilyNST20,HuangQYC21}, where the data $x$ is multiplied by a random matrix $A$ to create a small-space ``sketch'' $Ax$ of the original dataset. 
Note that sampling can fall under the umbrella of linear sketching in the case where the random matrix only contains a single one as the nonzero entry in each row. 
Unfortunately, linear sketching again cannot handle the implicit deletions of the sliding window model, since it is not entirely clear how to ``undo'' the effect of each expired element in the linear sketch $Ax$. 

\paragraph{Adapting insertion-only streaming algorithms to the sliding window model. }Algorithms for the sliding window model are often adapted from the insertion-only streaming model through either the exponential histogram framework~\cite{DatarGIM02} or its generalization, the smooth histogram framework~\cite{BravermanO07}. 
These frameworks transform streaming algorithms for either an additive function (in the case of exponential histograms) or a smooth function (in the case of smooth histograms) into sliding window algorithms by maintaining a logarithmic number of instances of the streaming algorithm, starting at various timestamps during the stream. 
Informally, a function is smooth if once a suffix of a data stream becomes a $(1+\beta)$-approximation of the entire data stream for the function, then the suffix is always a $(1+\alpha)$-approximation, regardless of the subsequent updates in the stream. 
Thus at the end of the stream of say length $m$, two of the timestamps must ``sandwich'' the beginning of the window, i.e., there exists timestamps $t_1$ and $t_2$ such that $t_1\le m-W+1<t_2$. 
The main point of the smooth histogram is that the streaming algorithm starting at time $t_1$ must output a value that is a good approximation of the function on the sliding window due to the smoothness of the function. 
Therefore, the smooth histogram is a cornerstone of algorithmic design in the sliding window model and handles many interesting functions, such as $L_p$ norm estimation (and in particular the sum), longest increasing subsequence, geometric mean, distinct elements estimation, and counting the frequency of a specific item. 

On the other hand, there remain interesting functions that are not smooth, such as clustering~\cite{BravermanLLM16,BorassiELVZ20,EpastoMMZ21}, submodular optimization~\cite{ChenNZ16,EpastoLVZ17}, sampling~\cite{JayaramWZ22}, regression and low-rank approximation~\cite{BravermanDMMUWZ20,UpadhyayU21}, and crucially for our purposes, heavy hitters~\cite{BravermanGO14,BravermanGLWZ18,Upadhyay19,WoodruffZ21}. 
These problems cannot be handled by the smooth histogram framework and thus for these problems, sliding windows algorithms were developed utilizing the specific properties of the objective functions. 

\paragraph{Previous work in the DP setting.} 
Among the previous literature, the work most related to the subject of our study is \cite{Upadhyay19} who proposed the study of differentially private $L_1$-heavy hitter algorithms in the sliding window. 
Although \cite{Upadhyay19} gave a continual release algorithm, which was later improved by~\cite{HuangQYC21}, the central focus of our work is the ``one-shot'' setting, where the algorithm releases a single set of statistics at the end of the stream, because permitting a single interaction with the data structure can often achieve better guarantees for both the space complexity and the utility of the algorithm. 
Indeed, in this paper we present $L_2$-heavy hitter algorithms for both the continual release and the one-shot settings, but the space/accuracy tradeoffs in the latter are much better than the former. 
\cite{Upadhyay19} also proposed a ``one-shot'' algorithm, which empirically performs well, but lacks the theoretical guarantees claimed in the paper. 
We refer to \secref{sec:overview} for more details. 

Privately releasing heavy-hitters in other big data models has also received significant attention. 
\cite{DworkNPRY10} introduced the problem of $L_1$-heavy hitters and other problems in the \emph{pan-privacy} streaming model, where the goal is to preserves differential privacy even if the internal memory of the algorithm is compromised, while \cite{ChanLSX12} considered the problem of continually releasing $L_1$-heavy hitters in a stream. 
The heavy-hitter problem has also been extensively studied in the local model~\cite{BassilyS15,DingKY17,AcharyaS19,BunNS19,BassilyNST20}. 
In the local model, individual users locally add privacy to their data, e.g., through randomized response, before sending their private information to a central and possibly untrusted server to aggregate the statistics across all users. 

\subsection{Overview of Our Techniques}\seclab{sec:overview}
We first use the smooth histogram to obtain a constant factor approximation to the $L_2$ norm of the sliding window similar to existing heavy-hitter non-DP algorithms in the sliding window model~\cite{BravermanGO14,BravermanGLWZ18}. 
We maintain a series of timestamps $t_1<t_2<\ldots<t_s$ for $s=\O{\log n}$, such that $L_2(t_1:m)>L_2(t_2:m)>\ldots>L_2(t_s:m)$ and $t_1\le m-W+1<t_2$. 
Hence, $L_2(t_1:m)$ is a constant factor approximation to $L_2(m-W+1:m)$, which is the $L_2$ norm of the sliding window. 
For each timestamp $t_i$ with $i\in[s]$, we also run an $L_2$-heavy hitter algorithm $\countsketch_i$, which outputs a list $\calL_i$ of size at most $\O{\frac{1}{\alpha^2}}$ that contains the $L_2$-heavy hitters of the suffix of the stream starting at time $t_i$, as well as approximations to each of their frequencies. 
It might be tempting to simply output a noisy version of the list $\calL_1$ output by $\countsketch_1$, since $t_1$ and $t_2$ sandwich the start of the sliding window, $m-W+1$. 
Indeed, this is the approach by~\cite{Upadhyay19}, although they only consider the $L_1$-heavy hitter algorithm $\countmin$ because they study the weaker $L_1$-heavy hitter problem and they do not need to run a norm estimation algorithm because $L_1$ can be computed exactly. 
However, \cite{BravermanGO14,BravermanGLWZ18} crucially note that $\calL_1$ can also include a number of items that are heavy-hitters with respect to the suffix of the stream starting at time $t_1$ but \emph{are not} heavy-hitters in the sliding window because many or even all of them appeared before time $m-W+1$. 
Thus although $\calL_1$ can guarantee that all the $L_2$-heavy hitters are reported by considering a lower threshold, say $\frac{\alpha}{2}$, the frequencies of each reported heavy-hitter can be arbitrarily inaccurate.

Observe it does not suffice to instead report the $L_2$-heavy hitters starting from time $t_2$. 
Although this will remove the false-positive issue of outputting items that are not heavy-hitters, there is now a false-negative issue; there may be heavy-hitters that appear after time $m-W+1$ but before time $t_2$ that will not be detected by $\countsketch_2$. 
Hence, there may be heavy-hitters of the sliding window that are not reported by $\calL_2$. 
See \figref{fig:sliding} for an example. 

\begin{figure*}[!htb]
\centering
\begin{tikzpicture}[scale=1] 
\draw [->] (0,0.15) -- (10,0.15);
\node at (-0.75,0.15){Stream:};
\draw (6.5,0) rectangle+(10-6.5,0.3);
\draw [decorate, decoration = {brace}] (10,-0.1) -- (6.5,-0.1);
\node at (8.5,-0.5){Active elements (sliding window)};

\filldraw[shading=radial, inner color = white, outer color = gray!50!, opacity=1] (0,1*0.3) rectangle+(10,0.3);
\filldraw[shading=radial, inner color = white, outer color = gray!50!, opacity=1] (3.3,2*0.3) rectangle+(10-3.3,0.3);
\filldraw[shading=radial, inner color = white, outer color = blue!50!, opacity=1] (5.6,3*0.3) rectangle+(10-5.6,0.3);
\filldraw[shading=radial, inner color = white, outer color = gray!50!, opacity=1] (7,4*0.3) rectangle+(10-7,0.3);
\filldraw[shading=radial, inner color = white, outer color = gray!50!, opacity=1] (9,5*0.3) rectangle+(10-9,0.3);

\end{tikzpicture}
\caption{Informally, we start a logarithmic number of streaming algorithms (the grey rectangles) at different points in time. We call the algorithm with the shortest substream that contains the active elements at the end of the stream (the blue rectangle). The challenge is that there may be heavy-hitters with respect to the blue rectangle that only appear before the active elements and therefore may be detected as heavy-hitters of the sliding window even though they are not.}
\figlab{fig:sliding}
\end{figure*}
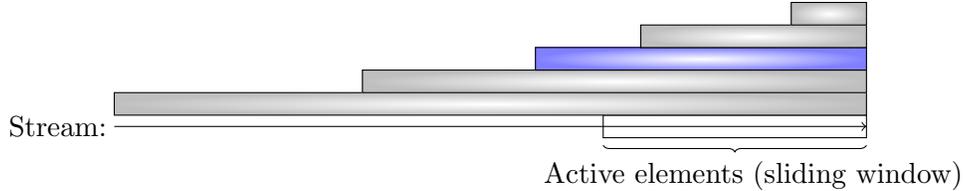

\paragraph{Approximate counters.} 
The fix by \cite{BravermanGO14,BravermanGLWZ18} that is missed by \cite{Upadhyay19} is to run approximate counters for each item $k\in[n]$ reported by some heavy-hitter algorithm $\countsketch_i$, i.e., there exists $i\in[s]$ such that $k\in\calL_i$. 
An approximate counter is simply a sliding window algorithm that reports a constant factor approximation to the frequency of a specific item $k\in[n]$. 
One way to achieve an approximate counter is to use the smooth histogram framework~\cite{BravermanO07}, but we show that an improved accuracy can be guaranteed if the maintenance procedure instead considers additive error rather than multiplicative error. 
Given the approximate counter that reports an estimate $\widehat{f_k}$ as the frequency for an item $k\in[n]$, we can then compare $\widehat{f_k}$ to the estimated $L_2$ norm of the sliding window to determine whether $k$ could possibly be an $L_2$-heavy hitter. 
This rules out the false positives that can be returned in $\calL_1$ without incurring false negatives omitted by $\calL_2$. 

\paragraph{Large sensitivity of subroutines.} 
So far we have only discussed the techniques required to release $L_2$-heavy hitters in the non-DP setting. In order to achieve differential privacy, a first attempt might be to add Laplacian noise to each of the procedures. 
Namely, we would like to add Laplacian noise to the estimate of the $L_2$ norm of the sliding window and the frequency of each reported heavy-hitter. 
However, since both the estimate of the $L_2$ norm of the sliding window and the frequency of each reported heavy-hitter is governed by the timestamps $t_1,\ldots,t_s$, then the sensitivity of each quantity can be rather large. 
In fact, if the frequency of each reported heavy-hitter has sensitivity $\alpha\cdot L_2(m-W+1:m)$ through the approximate counters, then with high probability, the Laplacian noise added to the frequency of some reported heavy-hitter will completely dominate the actual frequency of the item to the point where it is no longer possible to identify the heavy-hitters. 
Thus the approximate counters missed by~\cite{Upadhyay19} actually pose a significant barrier to the privacy analysis of the algorithm when each update can increase the value of a coordinate of the underlying vector by more than a unit amount, though this is a non-issue when all updates are uniform increments. 

\paragraph{Noisy timestamps.}
Instead of adding large Laplacian noise to each of the estimates, another possible attempt might be to make the timestamps in the histogram themselves noisy, e.g., by adding Laplacian noise to each of the timestamps. 
At first, it seems that the timestamps crucially govern the approximation guarantees by the smooth histogram and so adding noise would disrupt any sort of quality-of-approximation guarantee. 
However, upon closer examination, one can observe that due to the properties of $L_2$ and the count of an item, the Laplacian noise added to a timestamp would induce only a small additive error on each of the estimations. 
Unfortunately, we would no longer have sketches that correspond to the noisy timestamps. 
That is, suppose the smooth histogram maintains a heavy-hitter algorithm $\countsketch_1$ starting at a time $t_1$. 
Prior to releasing the statistics, suppose we add noise to the value of $t_1$ and obtain a noisy timestamp $\tilde{t_1}$. 
We would like to release the statistics of the dataset that begins with the $\tilde{t_1}$-th update of the stream, but it is not clear how to do so because we do not actually have a streaming algorithm starting at a time $\tilde{t_1}$. 
We could use $\countsketch_1$ as a proxy but that defeats the purpose of adding noise to the timestamp in the first place.

\paragraph{Lower smooth sensitivity through better approximations.}
Instead, we adapt the techniques of~\cite{BlockiGMZ22}, which provides a general analysis for the differential privacy of sliding window algorithms through smooth sensitivity~\cite{NissimRS07}. 
The main idea of \cite{BlockiGMZ22} is the following --- given an $\alpha$-approximation algorithm $\calA$ for a function with sensitivity $\Delta_f$, we would like to intuitively say the approximation algorithm has sensitivity $\alpha\Delta_f$. 
Unfortunately, this is not true because $\calA(X)$ may report $\alpha\cdot f(X)$ and $\calA(Y)$ may report $\frac{1}{\alpha}\cdot f(Y)$ for adjacent datasets $X$ and $Y$. 
However, if $\calA$ is instead a $(1+\alpha)$-approximation algorithm, then difference of the output of $\calA$ on $X$ and $Y$ can be bounded by $\alpha\cdot f(X)+\alpha\cdot f(Y)+\Delta_f$ through a simple triangle inequality, \emph{conditioned on the correctness} of $\calA$. 
In other words, if $\alpha$ is sufficiently small, then \cite{BlockiGMZ22} showed that the \emph{local sensitivity} of $\calA$ is sufficiently small, which allows control over the amount of Laplacian noise that must be added through existing mechanisms for smooth sensitivity. 
Unfortunately, if $\calA$ is not correct, then even the local sensitivity could be quite large; \cite{BlockiGMZ22} handled these cases separately by analyzing the smooth sensitivity of an approximation algorithm that is always correct and then arguing indistinguishability through statistical distance. 
Therefore generalizing the approach of \cite{BlockiGMZ22}, we can set the accuracy of the $L_2$ norm estimation algorithm, each $L_2$-heavy hitter algorithm, and each approximate counter algorithm to be sufficiently small and finally we can add Laplacian noise to each procedure without significantly impacting the final check of whether the estimated frequency for each item exceeds the heavy-hitter threshold. 

\paragraph{Pure differential privacy for $L_1$-heavy hitters in the sliding window model.}
Due to the linearity of $L_1$, our algorithm for differentially private $L_1$-heavy hitters in the sliding window model is significantly simpler than the $L_2$-heavy hitters algorithm. 
For starters, each set of $c$ updates must contribute exactly $c$ to the $L_1$ norm, whereas their contribution to the $L_2$ norm depends on the particular coordinates they update. 
Therefore, not only do we not require an algorithm to approximate the $L_1$ norm of the active elements of the sliding window, but also we can fix a set of static timestamps in the smooth histogram, so we do not need to perform the same analysis to circumvent the sensitivity of the timestamps. 
Instead, it suffices to initialize a \emph{deterministic} $L_1$-heavy hitter algorithm at each timestamp and maintain deterministic counters for each reported heavy-hitter. 
Pure differential privacy then follows from the lack of failure conditions in the subroutines, which was not possible for $L_2$-heavy hitters. 

\subsection{Paper Organization}
We present preliminaries on differential privacy, norm estimation, heavy-hitter algorithms, and the sliding window model in \secref{sec:prelims}. 
In \secref{sec:dp:hh:ltwo:sw}, we give a private algorithm for $L_2$-heavy hitters in the sliding window model. 
We then show in \secref{sec:lone} how our main technical ideas can be simplified and applied to give a $L_1$-heavy hitter algorithm in the sliding window model. 
Finally, we brief mention how to achieve continual release of $L_1$ heavy-hitters in \appref{app:continual:lone}. 

\section{Preliminaries}
\seclab{sec:prelims}
For an integer $n>0$, we use the notation $[n]:=\{1,\ldots,n\}$. 
We use the notation $\poly(n)$ to represent a constant degree polynomial in $n$ and we say an event occurs \emph{with high probability} if the event holds with probability $1-\frac{1}{\poly(n)}$. 
We say that $\calA$ is an $(\alpha,\delta)$-approximation algorithm for the function $f:\calU^*\to\mathbb{R}$ if for any $X\in\calU^*$, we have that 
\[\PPr{(1-\alpha)f(X)\le\calA(X)\le(1+\alpha)f(X)}\ge 1-\delta.\] 

\subsection{Differential Privacy}
In this section, we first introduce simple or well-known results from differential privacy. 
We say that streams $\frakS$ and $\frakS'$ are \emph{neighboring}, if there exists a single update $i\in[m]$ such that $u_i\neq u'_i$, where $u_1,\ldots,u_m$ are the updates of $\frakS$ and $u'_1,\ldots,u'_m$ are the updates of $\frakS'$.

\begin{definition}[$L_1$ sensitivity]
The \emph{$L_1$ sensitivity} of a function $f:\calU^*\to\mathbb{R}^k$ is defined by
\[\Delta_f=\max_{x,y\in \calU^*|,\|x-y\|_1=1}\|f(x)-f(y)\|_1.\]
\end{definition}
The $L_1$ sensitivity of a function $f$ bounds the amount that $f$ can change when a single coordinate of the input to $f$ changes and is often used to parameterize the amount of added noise to ensure differential privacy. 
For example, random noise may be generated from the Laplacian distribution:
\begin{definition}[Laplace distribution]
We say a random variable $X$ is \emph{drawn from a Laplace distribution with mean $\mu$ and scale $b>0$} if the probability density function of $X$ at $x$ is $\frac{1}{2b}\exp\left(-\frac{|x-\mu|}{b}\right)$. 
We use the notation $X\sim\Lap(b)$ to denote that $X$ is drawn from the Laplace distribution with scale $b$ and mean $\mu=0$. 
\end{definition}

 \begin{fact}\factlab{fact:lap}
 If $Y \sim \Lap(b)$, then  $\Pr[ \vert Y \vert \geq \ell \cdot b] = \exp(-\ell)$.
 \end{fact} 

In particular, the Laplace mechanism adds Laplacian noise with scale $\Delta_f$, the $L_1$ sensitivity of the function $f$. 
\begin{definition}[Laplace mechanism]
\deflab{def:lap:mech}
Given a function $f:\calU^*\to\mathbb{R}^k$, the \emph{Laplace mechanism} is defined by:
\[\calM_L(x,f,\eps)=f(x)+(X_1,\ldots,X_k),\]
where $X_i\sim\Lap(\Delta_f/\eps)$ for $1\leq i\leq k$.
\end{definition}
The Laplace mechanism is one of the most common methods of guaranteeing pure differential privacy. 
\begin{theorem}
[{\cite{DworkR14}}]
\thmlab{thm:dp:laplace}
The Laplace mechanism preserves $(\eps,0)$-differential privacy when $\Delta_f$ is the $L_1$ sensitivity. 
\end{theorem}
We define the following notion of local $L_1$ sensitivity for a fixed input, which can be much smaller than the (global) $L_1$ sensitivity. 
\begin{definition}[Local sensitivity]
For $f:\calU^*\to\mathbb{R}$ and $x\in\calU^*$, the \emph{local sensitivity of $f$ at $x$} is defined as

\[LS_f(x)=\max_{y:\|x-y\|_1=1}\|f(x)-f(y)\|_1.\]
\end{definition}
Unfortunately, the local sensitivity can behave wildly for specific algorithms. 
Thus we have the following definition that smooths such behavior for local sensitivity. 
\begin{definition}[Smooth upper bound on local sensitivity]\deflab{def:smooth-ub}
For $\beta>0$, a function $S:\calU^*\to\mathbb{R}$ is a \emph{$\beta$-smooth upper bound on the local sensitivity of $f:\calU^*\to\mathbb{R}$} if
\begin{enumerate}
\item \label{it:smooth-1}
For all $x\in\calU^*$, we have $S(x)\ge LS_f(x)$.
\item \label{it:smooth-2}
For all $x,y\in\calU^*$ with $\|x-y\|_1=1$, we have $S(x)\le e^\beta\cdot S(y)$. 
\end{enumerate}
\end{definition}
Even though the local sensitivity can be much smaller than the global $L_1$ sensitivity, the Laplace mechanism as defined in \defref{def:lap:mech} adds noise scaling with the global $L_1$ sensitivity. 
Hence it seems natural to hope for a mechanism that adds less noise. 
The following result shows that this is indeed possible. 
\begin{theorem}[Corollary 2.4 in \cite{NissimRS07}]
\thmlab{thm:smooth:laplace}
Let $f:\calU^*\to\mathbb{R}$ and $S:\calU^*\to\mathbb{R}$ be a $\beta$-smooth upper bound on the local sensitivity of $f$. 
If $\beta\le\frac{\eps}{2\ln(2/\delta)}$ and $\delta\in(0,1)$, then the mechanism that outputs $f(x)+X$, where $X\sim\Lap\left(\frac{2S(x)}{\eps}\right)$ is $(\eps,\delta')$-differentially private, for $\delta'=\frac{\delta}{2}\left(1+\exp\left(\frac{\eps}{2}\right)\right)$. 
\end{theorem}

%
We have the following theorems on the composition and post-processing of differentially private mechanisms. 
\begin{theorem}[Composition and post-processing of differential privacy~{\cite{DworkR14}}]
\thmlab{thm:dp:comp}
Let $\calM_i:\calU^*_i\to X_i$ be an $(\eps_i,\delta_i)$-differential private algorithm for $i\in[k]$. 
Then $\calM_{[k]}(x)=(\calM_1(x),\ldots,\calM_k(x))$ is $\left(\sum_{i=1}^k\eps_i,\sum_{i=1}^k\delta_i\right)$-differentially private. 
Moreover, if $g_i:X_i\to X'_i$ is an arbitrary random mapping, then $g_i(\calM_i(x))$ is $(\eps_i,\delta_i)$-differentially private. 
\end{theorem}

\begin{theorem}[Advanced composition of differential privacy~{\cite{DworkR14}}]
\thmlab{thm:dp:adv:comp}
For all $\eps,\delta\ge 0$ and $\delta'>0$, the advanced composition of $k$ algorithms, each of which is $(\eps,\delta)$-differentially private, is $(\tilde{\eps},\tilde{\delta})$-differentially private, where
\[\tilde{\eps}=\eps\sqrt{2k\ln(1/\delta')}+k\eps\left(\frac{e^\eps-1}{e^{\eps}+1}\right),\qquad \tilde{\delta}=k\delta+\delta'.\]
\end{theorem}


\subsection{Norm Estimation} 
\seclab{sec:norm}
In this section, we introduce preliminaries for norm or frequency moment estimation. 
\begin{definition}[Norm/moment estimation]
Given $p>0$ and a frequency vector $f\in\mathbb{R}^n$, we define the \emph{moment of $f$} by $F_p(f)=\sum_{k=0}^n f_k^p$ and the \emph{$L_p$ norm of $f$} by $\|f\|_p=L_p(f)=(F_p(f))^{1/p}$. 
For an accuracy parameter $\alpha\in(0,1)$, the $F_p$ \emph{moment estimation problem} is to output an estimated moment $\widehat{F}$ such that $|\widehat{F}-F_p(f)|\le\alpha F_p(f)$ and the \emph{norm estimation problem} is to output an estimated norm $\widehat{L}$ such that $|\widehat{L}-L_p(f)|\le\alpha L_p(f)$. 
\end{definition}
We note that $L_p$ is not a norm for $p\in(0,1)$, but the problem is nevertheless well-defined and also well-motivated due to the importance of frequency moment estimation for $p\in(0,1)$. 
Specifically, the norm and moment estimation problems are often used interchangeably because a $(1+\alpha)$-approximation to $L_p$ also gives a $(1+\alpha)^p=(1+\O{\alpha})$-approximation to $F_p$, for sufficiently small $\alpha$.  
Thus an algorithm that achieves $(1+\alpha)$-approximation to $L_p$ given an accuracy parameter $\alpha>0$ can also be adjusted to achieve a $(1+\alpha)$-approximation to $F_p$ by scaling the input accuracy $\alpha$. 

\begin{theorem}[Norm estimation algorithm AMS~{\cite{AlonMS99,BravermanCIW16, BlasiokDN17}}]
\thmlab{thm:ams}
Given an accuracy parameter $\alpha>0$ and a failure probability $\delta\in(0,1)$, there exists a one-pass streaming algorithm $\ams$ for the $L_2$ norm estimation problem, using $\O{\frac{\log n+\log m}{\alpha^2}\log\frac{\log m}{\alpha\delta}}$ bits of space and $\O{\frac{1}{\alpha^2}\log\frac{\log m}{\alpha\delta}}$ operations per time. 
\end{theorem}

We first recall the $\ams$ algorithm for $F_2$ frequency estimation~\cite{AlonMS99}, which is formalized in \algref{alg:ams} for constant probability of success. 
The algorithm outputs a $(1+\alpha)$-approximation to the second moment of the frequency vector using $\O{\frac{1}{\alpha^2}\log n\log\frac{1}{\delta}}$ space. 
The algorithm first generates a random sign vector $s$ of length $n$, so that $s_i\in\{-1,+1\}$ for each $i\in[n]$. 
The algorithm then computes the inner product $Z=\langle s,f\rangle$, so that $Z^2$ is an unbiased estimator of $F_2$ with variance $\O{F_2^2}$. 
It follows from a standard variance reduction argument through Chebyshev's inequality that the mean of $\O{\frac{1}{\alpha^2}}$ such inner products is a $(1\pm\alpha)$ approximation of $F_2$ with constant probability and thus the median of $\O{\log\frac{1}{\delta}}$ such means is a $(1\pm\alpha)$ approximation of $F_2$ with probability at least $1-\delta$. 
Thus by taking the median of $\O{\log\frac{1}{\delta}}$ independent instances of \algref{alg:ams} boosts the probability of success to $1-\delta$. 
Similarly, the square root of the median is a $(1\pm\alpha)$ approximation of $L_2$ with probability at least $1-\delta$, though one could also view the $L_2$ estimation algorithm as taking the median of means of the absolute value of each inner product $Z$, which is equivalent to taking the square root of $Z^2$. 

\begin{algorithm}[!htb]
\caption{Algorithm $\ams$ for $F_2$ estimation}
\alglab{alg:ams}
\begin{algorithmic}[1]
\Require{Stream $\frakS$, threshold/accuracy parameter $\alpha\in(0,1)$}
\Ensure{$(1+\alpha)$-approximation of $F_2$ with probability at least $2/3$}
\State{Set $r=\O{\frac{1}{\alpha^2}}$}
\State{Generate $r$ random sign vectors $s^{(1)},\ldots,s^{(r)}$}
\State{Initialize sums $S_1,\ldots,S_r=0$}
\For{each update $u_i\in[n]$, $i\in[m]$}
\For{each $j\in[r]$}
\State{Update $S_j\gets S_j+s^{(j)}_{u_i}$}
\EndFor
\EndFor
\State{\Return $\frac{1}{r}\sum_{j=1}^r(S_j)^2$}
\end{algorithmic}
\end{algorithm}

We first show the global $L_1$ sensitivity of the $L_2$ norm. 
\begin{lemma}[Sensitivity of $L_2$]
\lemlab{lem:sens:ltwo}
Let $f$ and $f'$ be frequency vectors on a universe of size $n$ such that $n-2$ coordinates of $f$ and $f'$ have the same value and the remaining two coordinates differ by exactly one. 
Then $|L_2(f)-L_2(f')|\le 2$. 
\end{lemma}
\begin{proof}
Observe that by the concavity of the square root function, we have $\sqrt{(x+1)^2+y}-\sqrt{x^2+y}\le 1$ for all $x,y\ge 0$. 
Let $i$ and $j$ be the indices where the coordinates of $f$ and $f'$ differ, i.e., $|f[i]-f'[i]|=1$, $|f[j]-f'[j]|=1$, and $f[k]=f'[k]$ for all $k\in[n]\setminus\{i,j\}$. 
Now we define $f''$ to be an intermediate frequency vector such that it differs by only one coordinate to $f$ and $f'$, respectively, i.e., $f''[i]=f[i]$, $f''[j]=f'[j]$, and $f''[k]=f[k]=f'[k]$ for all $k\in[n]\setminus\{i,j\}$. 
Then by the previous observation and the triangle inequality, we have $|L_2(f)-L_2(f')|\leq |L_2(f)-L_2(f'')| + |L_2(f'')-L_2(f')| \leq 1+1=2$.
\end{proof}

\subsection{Heavy Hitters}
In this section, we formally introduce the $L_p$-heavy hitter problem and the algorithm $\countsketch$, which is commonly used to find the $L_2$-heavy hitters. 
\begin{definition}[$L_p$-heavy hitter problem]
Given an accuracy/threshold parameter $\alpha\in(0,1)$, $p>0$, and a frequency vector $f\in\mathbb{R}^n$, report all coordinates $k\in[n]$ such that $f_k\ge\alpha\,L_p(f)$ and no coordinates $j\in[n]$ such that $f_j\le\frac{\alpha}{2}\,L_p(f)$. 
For each reported coordinate $k\in[n]$, also report an estimated frequency $\widehat{f_k}$ such that $|\widehat{f_k}-f_k|\le\frac{\alpha}{4}\,L_p(f)$.  
\end{definition}

\begin{theorem}[Heavy-hitter algorithm \countsketch~{\cite{CharikarCF04}}]
\thmlab{thm:countsketch}
Given an accuracy parameter $\alpha>0$ and a failure probability $\delta\in(0,1)$, there exists a one-pass streaming algorithm $\countsketch$ for the $L_2$-heavy hitter problem that uses $\O{\frac{1}{\alpha^2}\log\frac{n}{\delta}}$ words of space and $\O{\log\frac{n}{\delta}}$ update time. 
\end{theorem}

The $\countsketch$ data structure~\cite{CharikarCF04} is an $r$ by $b$ table, which can be thought of as $r$ rows of $b$ buckets, each with counters that are initialized to zero. 
In each row $j\in[r]$, each item of the universe $i\in[n]$ is assigned to one of the $b$ buckets by a hash function $h^{(j)}:[n]\to[b]$, so that the bucket for item $i$ is $h^{(j)}(i)$. 
If $i$ appears in the stream, then the random sign $s^{(j)}(i)$ is added to the counter corresponding to the bucket assigned to the item in row $j$ for each $j\in[r]$. 
At the end of the stream, the mean across all $r$ rows of the magnitude of the counters for the buckets assigned to $i$ corresponds to the estimate of the frequency of $i$. 
Due to the $\O{\log n}$ rows, the algorithm has failure probability $1-\frac{1}{\poly(n)}$. 
For completeness, we provide the full details in \algref{alg:cs}. 

\begin{algorithm}[!htb]
\caption{Algorithm $\countsketch$ for heavy-hitter estimation}
\alglab{alg:cs}
\begin{algorithmic}[1]
\Require{Stream $\frakS$, threshold/accuracy parameter $\alpha\in(0,1)$}
\Ensure{$L_2$-Heavy hitter algorithm}
\State{Set $r=\O{\log n}$, $b=\O{\frac{1}{\alpha^2}}$}
\State{Generate $r$ hash functions $h^{(1)},\ldots,h^{(r)}:[n]\to[b]$ and $s^{(1)},\ldots,s^{(r)}:[n]\to\{-1,+1\}$}
\State{Initialize sums $S_{i,j}=0$ for $(i,j)\in[r]\times[b]$}
\For{each update $u_i\in[n]$, $i\in[m]$}
\For{each $j\in[r]$}
\State{Set $b_{i,j}=h^{(j)}(u_i)$ and $s_{i,j}=s^{(j)}(u_i)$}
\State{Update $S_{j,b_{i,j}}=S_{j,b_{i,j}}+s_{i,j}$}
\EndFor
\EndFor
\For{each $i\in[n]$}
\State{Set $b_{i,j}=h^{(j)}(u_i)$ for each $j\in[r]$}
\State{\Return $\widehat{f_i}=\frac{1}{r}\sum_{j\in[r]}|S_{j,b_{i,j}}|$ as the estimated frequency for $f_i$}
\EndFor
\end{algorithmic}
\end{algorithm}

\subsection{Sliding Window Model}
In this section, we introduce simple or well-known results for the sliding window model. 
\begin{definition}[Sliding window model]
Given a universe $\calU$ of items, which we associate with $[n]$, let a stream $\frakS$ of length $m$ consist of updates $u_1,\ldots,u_m$ to the universe $\calU$, so that $u_i\in[n]$ for each $i\in[m]$. 
After the stream, a window parameter $W$ is given, which induces the frequency vector $f\in\mathbb{R}^n$ so that $f_k=|\{i: u_i=k\wedge i\ge m-W+1\}|$ for each $k\in[n]$. 
In other words, each coordinate $k$ of the frequency vector is the number of updates to $k$ within the last $W$ updates. 
\end{definition}

We say $A$ and $B$ are \emph{adjacent} substreams of a stream $\frakS$ of length $m$ if $A$ consists of the updates $u_i,\ldots,u_j$ and $B$ consists of the updates $u_{j+1},\ldots,u_k$ for some $i,j,k\in[m]$. 
We have the following definition of a smooth function for the purposes of sliding window algorithms, not to be confused with the smooth sensitivity definition for differential privacy. 
\begin{definition}[Smooth function]
Given adjacent substreams $A$ and $B$, a function $g:\calU^* \to\mathbb{R}$ is \emph{$(\alpha,\beta)$-smooth} if $(1-\beta)g(A\cup B)\le g(B)$ implies $(1-\alpha)g(A\cup B\cup C)\le g(B\cup C)$ for some parameters $0<\beta\le\alpha<1$ and \emph{any} adjacent substream $C$. 
\end{definition}

Smooth functions are a key building block in the smooth histogram framework by \cite{BravermanO10}, which creates a sliding window algorithm for a large number of functions using multiple instances of streaming algorithms starting at different points in time. 
See \algref{alg:hist} for more details on the smooth histogram. 

\begin{theorem}[Smooth histogram~{\cite{BravermanO10}}]\thmlab{thm:bo}
Given accuracy parameter $\alpha\in(0,1)$, failure probability $\delta\in(0,1)$ and an $(\alpha,\beta)$-smooth function $g:\calU^m\to\mathbb{R}$, suppose there exists an insertion-only streaming algorithm $\calA$ that outputs a $(1+\alpha)$-approximation to $g$ with high probability using space $\calS(\alpha,\delta,m,n)$ and update time $\calT(\alpha,\delta,m,n)$. 
Then there exists a sliding window algorithm that outputs a $(1+\alpha)$-approximation to $g$ with high probability using space $\O{\frac{1}{\beta}(\calS(\beta,\delta, m,n)+\log m)\log m}$ and update time $\O{\frac{1}{\beta}(\calT(\beta,\delta,m,n))\log m}$. 
\end{theorem}

\begin{algorithm}[!htb]
\caption{Smooth histogram~\cite{BravermanO10}}
\alglab{alg:hist}
\begin{algorithmic}[1]
\Require{Stream $\frakS$, accuracy parameter $\rho\in(0,1)$, streaming algorithm $\calA$ for $(\rho,\beta(\rho))$-smooth function}
\Ensure{$(1+\rho)$-approximation of predetermined function with probability at least $1-\delta$}
\State{$H\gets\emptyset$}
\For{each update $u_t$ with $t\in[m]$}
\State{$H\gets H\cup\{t\}$}
\For{each time $t_s\in H$}
\State{Let $x_s$ be the output of $\calA$ with failure probability $\frac{\delta}{\poly(n,m)}$ starting at time $t_s$ and ending at time $t$.}
\If{$x_{s-1}\le\left(1-\frac{\beta(\rho)}{2}\right)x_{s+1}$}
\State{Delete $t_s$ from $H$ and reorder the indices in $H$}
\EndIf
\EndFor
\EndFor
\State{Let $s$ be the smallest index such that $t_s\in H$ and $t_s\le m-W+1$.}
\State{Let $x_s$ be the output of $\calA$ starting at time $t_s$ at time $t$.}
\State{\Return $x_s$}
\end{algorithmic}
\end{algorithm}
The following smoothness parameters for the $F_p$ frequency moment and $L_p$ norm functions suffice for our purposes:
\begin{lemma}[{\cite{BravermanO10}}]
\lemlab{lem:fp:smooth}
For any $\rho\in(0,1)$, the $F_p$ and $L_p$ functions are $\left(\rho,\frac{\rho^p}{p}\right)$-smooth for $p\ge 1$ and $(\rho,\rho)$-smooth for $0<p\le 1$.  
\end{lemma}


\paragraph{Approximate frequency in the sliding window model.}
Finally, we present a deterministic algorithm $\counter$ that can be parametrized to give an additive $M$-approximation to the estimated frequency $\widehat{f_i}$ of a particular element $i\in[n]$ in the sliding window model. 
The algorithm initializes counters for $i$ starting at multiple times in the stream. 
To maintain additive error $M$, the algorithm removes a counter starting at a certain time, if the difference between the previous counter and the next counter is at most $M$, since it suffices to simply use the previous counter instead. 
For more details, see \algref{alg:counter}. 

\begin{algorithm}[!htb]
\caption{Algorithm $\counter$ for frequency estimation in the sliding window model}
\alglab{alg:counter}
\begin{algorithmic}[1]
\Require{Stream $\frakS$, window parameter $W>0$, additive error $M$, index $i\in[n]$}
\Ensure{$\widehat{f_i}$ such that $0\le f_i-\widehat{f_i}\le M$}
\State{$H\gets\emptyset$}
\For{each update $u_t\in[n]$ with $t\in[m]$}
\If{$u_t=i$}
\State{$H\gets H\cup\{t\}$}
\EndIf
\For{each time $t_s\in H$}
\State{Let $c_s$ be the number of instances of $i$ from $t_s$.}
\If{$c_{s-1}-c_{s+1}<M$}
\State{Delete $t_s$ from $H$ and re-index $H$}
\EndIf
\EndFor
\EndFor
\State{Let $s$ be the largest index such that $t_s\in H$ and $t_s>m-W+1$.}
\State{Let $\widehat{f_i}$ be the number of instances of $i$ from $t_s$.}
\State{\Return $\widehat{f_i}$}
\end{algorithmic}
\end{algorithm}

\begin{lemma}
\lemlab{lem:counter:one}
There exists a deterministic algorithm $\counter$ that outputs an additive $M$ approximation to the frequency of an element $i\in[n]$ in the sliding window model. 
The algorithm uses $\O{\frac{f_i}{M}\,\log m}$ bits of space. 
\end{lemma}


\section{Differentially Private Heavy-Hitters in the Sliding Window Model}
\seclab{sec:dp:hh:ltwo:sw}
In this section, we give a private algorithm for $L_2$-heavy hitters in the sliding window model. 
Our algorithm will initially use a smooth histogram approach by instantiating a number of $L_2$ norm estimation algorithm starting at various timestamps in the stream. 
Through a sandwiching argument, these $L_2$ norm estimation algorithms will provide a constant factor approximation to the $L_2$ norm of the sliding window, which will ultimately allow us to determine whether elements of the stream are heavy-hitters. 
Moreover, by using a somewhat standard smooth sensitivity argument, we can show that these subroutines can be maintained in a way that preserves differential privacy. 

To identify a subset of elements that can be heavy-hitters, we also run a private $L_2$-heavy hitters algorithm starting at each timestamp. 
Unfortunately, because the timestamps do not necessarily coincide with the beginning of the sliding window, it may be possible that depending on our approach, we may either output a number of elements with very low, possibly even zero, frequency, or we may neglect to output a number of heavy-hitters. 
To overcome this issue, we maintain private approximate counters for each item that is reported by our private $L_2$-heavy hitters algorithms, which allows us to rule out initially reported false positives without incurring false negatives. 

The crucial observation is all elements that are heavy-hitters with respect to the sliding window must first be reported by our heavy-hitter algorithm at some timestamp (not necessarily corresponding to the beginning of the sliding window). 
Although our private approximate counter for a universe element $i$ is initiated only after the heavy-hitter algorithm outputs element $i$, the approximate counter only misses a small number of the instances of $i$, because we run our heavy-hitter algorithm with a significantly lower threshold. 
Hence the approximate counter subroutine will still estimate the frequencies of the potential heavy-hitters elements with a sufficiently small error that is enough to differentiate whether an element is heavy with respect to the sliding window. 
We give the algorithm in full in \algref{alg:dp:sw:hh}. 

\begin{algorithm}[!htb]
\caption{Differentially private sliding window algorithm for $L_2$-heavy hitters}
\alglab{alg:dp:sw:hh}
\begin{algorithmic}[1]
\Require{Stream $\frakS$, accuracy parameter $\alpha\in(0,1)$, differential privacy parameters $\eps,\delta>0$, window parameter $W>0$, size $n$ of the underlying universe, upper bound $m$ on the stream length}
\Ensure{A list $\calL$ of $L_2$-heavy hitters with approximate frequencies}
\State{Process the stream $\frakS$, maintaining timestamps $t_1,\ldots,t_s$ at each time $t\in[m]$ so that for each $i\in[s]$, either $i=s$, $t_{i+1}=t_i+1$ or $L_2(t_i,t)\le\left(1+\left(\frac{\eps}{1000\log m}\right)^2\right)\,L_2(t_{i+1},t)$ through a smooth histogram with failure probability $\frac{\delta}{2m^2}$}
\State{Implement heavy-hitter algorithm $\countsketch$ on the substream starting at $t_i$ for each $i\in[s]$ with threshold $\frac{\alpha^3\eps}{500\log m}$ and failure probability $\frac{\delta}{2m^2}$}
\State{Set $a=\max\{i\in[s]:t_i\le m-W+1\}$ on window query $W>0$}
\State{Set $\widehat{L_2}$ to be an $\left(1+\frac{\eps}{500\log m}\right)$-approximation to $L_2(t_a,t)$ from the smooth histogram and $X\gets\Lap\left(\frac{1}{40\log m}\,\widehat{L_2}\right)$}
\For{each heavy-hitter $k\in[n]$ reported by $\countsketch$ starting at $t_a$}
\State{Run $\counter$ with additive error $\frac{\alpha^3\eps}{1000\log m}\widehat{L_2}$ for each reported heavy-hitter}
\State{Let $\widehat{f_k}$ be the approximate frequency reported by $\counter$}
\State{$Y_k\gets\Lap\left(\frac{\alpha}{75\log m}\,\widehat{L_2}\right)$, $Z_k\gets\Lap\left(\frac{\alpha}{75\log m}\,\widehat{L_2}\right)$, $\widetilde{f_k}=\widehat{f_k}+Z_k$}
\If{$\widetilde{f_k}\ge\frac{3\alpha}{4}\,(\widehat{L_2}+X)+Y_k$}
\State{$\calL\gets\calL\cup\{(k,\widetilde{f_k})\}$}
\EndIf
\EndFor
\State{\Return $\calL$}
\end{algorithmic}
\end{algorithm}

We first describe the procedure for the approximate frequency estimation for each reported heavy-hitter. 
Let $\countsketch_a$ be an $L_2$-heavy hitter algorithm starting at timestamp $t_a$, where $a=\max\{i\in[s]:t_i\le m-W+1\}$ on window query $W>0$. 
For a coordinate $k\in[n]$ that is reported by $\countsketch_a$ from times $t$ through $m$, we use $\counter$ to maintain a number of timestamps such that the frequency of $k$ on the suffixes induced by the timestamps are arithmetically increasing by roughly $\alpha^2\,L_2(f)/16$. We emphasize that we run the $\counter$ for each reported heavy-hitter in the same pass as the rest of the algorithm. 
\begin{lemma}
\lemlab{lem:counter}
Let $\calE$ be the event that (1) the smooth histogram data structure does not fail, (2) all instances of $\countsketch$ do not fail, and (3) $X\le\frac{L_2(f)}{10}$ and $\max_{j\in[n]}(Y_j,Z_j)\le\frac{\alpha L_2(f)}{10}$. 
Let $\countsketch_a$ be the instance of $\countsketch$ starting at time $t_a$. 
Conditioned on $\calE$, then for each reported heavy-hitter $k$ by $\countsketch_a$, \algref{alg:dp:sw:hh} outputs an estimated frequency $\widehat{f_k}$ such that
\[|f_k-\widehat{f_k}|\le\frac{\alpha^3\eps}{500\log n}\,L_2(f).\]
The algorithm uses $\O{\frac{1}{\alpha^6\eps^2}\log^3 m}$ space and $\O{\frac{\log^2 m}{\alpha^4\eps^2}}$ update time per instance of $\countsketch$. 
\end{lemma}
\begin{proof}
An approximate frequency counter for $k$ maintains a series of additional timestamps $t^{(k)}_1<\ldots<t^{(k)}_r$ that denotes suffixes of the stream, with the property that for each $i\in[r-2]$, the frequency of $k$ from time $t^{(k)}_i$ to $m$ is more than an additive $\frac{\alpha^3\eps}{500\log m}\,\widehat{L_2}$ amount than the frequency of $k$ from time $t^{(k)}_{i+2}$.  
Suppose without loss of generality that the stream has length $m$ and coordinate $k\in[n]$ is reported by $\countsketch_a$ from times $t$ through $m$ so that $t^{(k)}_1=t$. 
By a combination of casework on the time $t$ compared to $m-W+1$ and a standard sandwiching argument, we bound the error of the estimated frequency $\widehat{f_k}$ for $k$. 

We have two cases: either $t<m-W+1$ or $t\ge m-W+1$. 
First suppose that $t<m-W+1$ so that $t^{(k)}_1=t<m-W+1$. 
Then there exists an index $i\in[r-1]$ such that $t^{(k)}_i\le m-W+1<t^{(k)}_{i+1}$.  
Let $T_i$ be frequency of $k$ between times $t^{(k)}_i$ to $m$ and $T_{i+1}$ be frequency of $k$ between times $t^{(k)}_{i+1}$ to $m$, so that $T_{i+1}\le f_k\le T_i$. 
Then by \lemref{lem:counter:one} and the smoothness of $L_1$ of the frequency vector restricted to coordinate $k$ in \lemref{lem:fp:smooth}, we have $T_{i+1}\le T_i\le T_{i+1}+\frac{\alpha^3\eps}{1000\log m}\,\widehat{L_2}$. 
Thus, both $T_i$ and $T_{i+1}$ are additive $\frac{\alpha^3\eps}{1000\log m}\,\widehat{L_2}$-additive approximations to $f_k$, so that
\[|f_k-\widehat{f_k}|\le\frac{\alpha^3\eps}{1000\log m}\,\widehat{L_2}.\]
Conditioning on $\calE$, then by the smooth histogram and the smoothness of $L_2$, i.e., \lemref{lem:fp:smooth}, we have that $\widehat{L_2}(f)\le\left(1+\frac{\eps}{500\log m}\right)\,L_2(f)$, so it follows that 
\[|f_k-\widehat{f_k}|\le\frac{\alpha^3\eps}{500\log m}\,L_2(f).\]

On the other hand, suppose $t\ge m-W+1$. 
Then conditioning on $\calE$ (the correctness of $\countsketch_a$), at most $\frac{\alpha^3\eps}{500\log m}\,\widehat{L_2}$ instances of $k$ have arrived between times $t_a$ and $t$; otherwise item $k$ would have also been reported as an $L_2$-heavy hitter by $\countsketch_a$ with threshold $\frac{\alpha^3\eps}{500\log m}$ at time $t-1$. 
Since $t^{(k)}_1=t$, then the frequency $T_i$ of $k$ between times $t^{(k)}_i$ to $m$ is exactly the number of updates to $k$ since time $t$. 
Thus for $\widehat{f_k}=T_i$, we have that $|\widehat{f_k}-T_i|\le\frac{\alpha^3\eps}{500\log m}\,L_2(f)$. 

To analyze the space complexity, note that by the invariance of the timestamp maintenance, we have $T_i-T_{i+2}>\frac{\alpha^3\eps}{1000\log m}\,\widehat{L_2}$ for each index $i\in[r]$ corresponding with timestamps $t^{(k)}_1<\ldots<t^{(k)}_r$. 
Similarly, the same invariance holds for the timestamps corresponding to the counters for other heavy-hitters $j\in[n]$ reported by $\countsketch_a$. 
Since $\widehat{L_2}>\left(1+\frac{\eps}{500\log m}\right)\,L_2(f)$, then it follows that there can only be $\O{\frac{\log^2 m}{\alpha^3\eps^2}}$ timestamps. 
Each timestamp corresponds to a counter that is encoded using $\O{\log m}$ bits, so the counters use $\O{\frac{\log^3 m}{\alpha^3\eps^2}}$ space. 
On the other hand, by \thmref{thm:countsketch}, each instance of $\countsketch$ with threshold $\frac{\alpha^3\eps}{500\log m}$ uses $\O{\frac{1}{\alpha^6\eps^2}\log^4 m}$ bits of space, which gives the space complexity. 

To analyze the time complexity, note that each update corresponds to whether each of the counters corresponding to the $\O{\frac{\log^2 m}{\alpha^3\eps^2}}$ timestamps satisfies the invariant that $T_i-T_{i+2}>\frac{\alpha^3\eps}{1000\log m}\,\widehat{L_2}$. 
Thus the update time is $\O{\frac{\log^2 m}{\alpha^3\eps^2}}$ operations. 
\end{proof}

We first show that the list $\calL$ output by \algref{alg:dp:sw:hh} does not contain any items with ``low'' frequency. 
\begin{lemma}[Low frequency items are not reported]
\lemlab{lem:hh:low}
Let $\calE$ be the event that (1) the smooth histogram data structure does not fail, (2) all instances of $\countsketch$ do not fail, and (3) $X\le\frac{L_2(f)}{10}$ and $\max_{j\in[n]}(Y_j,Z_j)\le\frac{\alpha L_2(f)}{10}$. 
Let $f$ be the frequency vector induced by the sliding window parameter $W$ and suppose $f_k\le\frac{\alpha}{2}\,L_2(f)$. 
Then conditioned on $\calE$, $k\notin\calL$. 
\end{lemma}
\begin{proof}
Let $\countsketch_a$ be the instance of $\countsketch$ starting at time $t_a$. 
Let $H_a$ be the set of heavy-hitters reported by $\countsketch_a$ with threshold $\alpha/16$ at time $m$. 
Then conditioned on $\calE$, either (1) $k\notin H_a$ so that $k\notin\calL$ and thus $k$ will not be reported by \algref{alg:dp:sw:hh} or (2) $k\in H_a$, so that $k$ is an $\alpha/16$ heavy-hitter of some suffix of the stream. 
In the latter case, an approximate frequency estimate $\widehat{f_k}$ for $k$ is output by an instance of $\counter$. 

By \lemref{lem:counter}, we have that 
\[|f_k-\widehat{f_k}|\le\frac{\alpha^3\eps}{500\log m}\,L_2(f).\]
Thus, $\widehat{f_k}\le f_k+\alpha\,L_2(f)/8$. 
By the smooth histogram, we also have that $\left(1+\frac{\eps}{500\log m}\right)\widehat{L_2}\ge L_2(f)$. 
Hence $f_k\le\frac{\alpha}{2}\,L_2(f)$ along with the assumption that $X\le\frac{L_2(f)}{10}$ and $\max_{j\in[n]}(Y_j,Z_j)\le\frac{\alpha L_2(f)}{10}$ conditioned on $\calE$ implies that $\widehat{f_k}<\frac{3\alpha}{4}\,(\widehat{L_2}+X)+Y_k$ so $k$ will not be added to $\calL$ and thus not reported by \algref{alg:dp:sw:hh}. 
\end{proof}

We now show that the heavy-hitters are reported and bound the error in the estimated frequency for each reported item.  
\begin{lemma}[Heavy-hitters are estimated accurately]
\lemlab{lem:hh:high}
Let $f$ be the frequency vector induced by the sliding window parameter $W$. 
Let $\calE$ be the event that (1) the smooth histogram data structure does not fail, (2) all instances of $\countsketch$ do not fail, and (3) $X\le\frac{L_2(f)}{10}$ and $\max_{j\in[n]}(Y_j,Z_j)\le\frac{\alpha L_2(f)}{10}$. 
Conditioned on $\calE$, then $k\in\calL$ for each $k\in[n]$ with $f_k\ge\alpha\,L_2(f)$. 
Moreover, for each item $k\in\calL$, 
\[|f_k-\widehat{f_k}|\le\frac{\alpha^3\eps}{500\log m}\,L_2(f).\]
\end{lemma}
\begin{proof}
Let $\countsketch_a$ be the instance of $\countsketch$ starting at time $t_a$. 
Since $f_k\ge\alpha\,L_2(f)$ and $L_2(t_a,t)\le\left(1+\frac{\eps}{500\log m}\right)\,L_2(f)$ implies $f_k\ge\frac{\alpha}{16}\,L_2(t_a,t)$, then $f_k$ will be reported by $\countsketch_a$, conditioned on $\calE$ (the correctness of $\countsketch_a$). 
We thus have an approximate frequency estimate $\widehat{f_k}$ for $k$ output by an instance of $\counter$. 
By \lemref{lem:counter}, 
\[|f_k-\widehat{f_k}|\le\frac{\alpha^3\eps}{500\log m}\,L_2(f).\]
Furthermore, we have that $\widehat{L_2}\le\left(1+\frac{\eps}{500\log m}\right)\,L_2(t_a,t)\le\left(1+\frac{\eps}{500\log m}\right)^2\,L_2(f)$. 
Given the assumption that $X\le\frac{L_2(f)}{10}$ and $\max_{j\in[n]}(Y_j,Z_j)\le\frac{\alpha L_2(f)}{10}$ conditioned on $\calE$, then $\widehat{f_k}<\frac{3\alpha}{4}\,(\widehat{L_2}+X)+Y_k$ so $k\in\calL$. 

Let $H_a$ be the set of heavy-hitters reported by $\countsketch_a$ with threshold $\frac{\alpha^3\eps}{500\log m}$ at time $m$. 
For each $j\in\calL$, we have that $j\in H_a$ so that an approximate frequency estimate $\widehat{f_j}$ for $j$ is output by an instance of $\counter$. 
By \lemref{lem:counter}, we have that $|f_j-\widehat{f_j}|\le\frac{\alpha^3\eps}{500\log m}\,L_2(f)$, as desired. 
\end{proof}
We now show that the event $\calE$ conditioned by \lemref{lem:counter}, \lemref{lem:hh:low}, and \lemref{lem:hh:high} occurs with high probability.

\begin{claim}\claimlab{claim:lap-HH}
Let $\PPr{\widehat{L_2}(f) \leq 2 L_2(f)}\leq \delta$.  Suppose $X \sim \Lap \left(\frac{1}{40 \log m} \widehat{L_2}(f)\right)$, and for $j \in [n]$, $Y_j, Z_j \sim \Lap \left(\frac{\alpha }{75 \log m} \widehat{L_2}(f)\right)$. 
Then, 
$\PPr{X\le\frac{L_2(f)}{10}\,\wedge\,\max_{j\in[n]}(Y_j,Z_j)\le\frac{\alpha L_2(f)}{10}}\ge 1-\left(\frac{1}{m^2}+\frac{2}{m^{\frac{11}{4}}}+\delta\right)$. 
\end{claim}
\begin{proof}
Using \factref{fact:lap},
    \begin{align*}
        \PPr{\vert X \vert > \frac{\widehat{L_2}(f)}{20} } 
        = \PPr{\vert X \vert > \frac{\widehat{L_2}(f)}{40 \log m}\cdot 2 \log m } = \frac{1}{m^2}
    \end{align*}

Also using \factref{fact:lap}, for a fixed $j$, 
    \begin{align*}
        \PPr{\vert Y_j \vert > \frac{\alpha \widehat{L_2}(f)}{20} } 
        = \PPr{\vert Y_j \vert > \frac{\widehat{L_2}(f)}{75 \log m}\cdot \frac{75 }{20}\log m } = \frac{1}{m^{\frac{15}{4}}}
    \end{align*}
    By a union bound over all $j\in [n]$, 
    \begin{align*}
        \PPr{\max_{j\in [n]}\vert Y_j \vert > \frac{\alpha \widehat{L_2}(f)}{20} } \leq \frac{n}{m^{\frac{15}{4}}}
    \end{align*}
    Since $Z_j$ is identically distributed and over $j \in [n]$, and $n<m$, we have that
    \begin{align*}
        \PPr{\max_{j\in [n]}\vert (Y_j,Z_j) \vert > \frac{\alpha \widehat{L_2}(f)}{20}} \leq \frac{2n}{m^{\frac{15}{4}}}\leq \frac{2}{m^{\frac{11}{4}}} 
    \end{align*}
    We know that with probability $1-\delta$, $\widehat{L_2}(f) \leq 2 L_2(f)$, therefore, by a union bound, and using $n<m$, 
    \begin{align*}
        \PPr{(\widehat{L_2}(f)>2L_2(f))\vee \left(\max_{j\in [n]}\vert (Y_j,Z_j) \vert > \frac{\alpha \widehat{L_2}(f)}{20}\right) \vee \left(\vert X \vert > \frac{1}{20}\widehat{L_2}(f) \right)  } &\leq \frac{2}{m^{\frac{11}{4}}}+\frac{1}{m^2}+\delta.
    \end{align*}   

\end{proof}
\begin{lemma}
\lemlab{lem:fail:prob}
Let $\calE$ be the event that (1) the smooth histogram data structure does not fail on either stream, (2) all instances of $\countsketch$ do not fail, and (3) $X\le\frac{L_2(f)}{10}$ and $\max_{j\in[n]}(Y_j,Z_j)\le\frac{\alpha L_2(f)}{10}$. 
Then $\PPr{\calE}\ge 1-\frac{4}{m^2}-\frac{2}{m^{\frac{11}{4}}}$.
\end{lemma}
\begin{proof}
Let $\calE_1$ be the event that the smooth histogram data structure does not fail, $\calE_2$ be the event that all instances of $\countsketch$ do not fail, and $\calE_3$ be the event that $X\le\frac{L_2(f)}{10}$ and $\max_{j\in[n]}(Y_j,Z_j)\le\frac{\alpha L_2(f)}{10}$, so that $\calE=\calE_1\wedge\calE_2\wedge\calE_3$. 
It suffices to set the probability of failure $\delta=\frac{1}{m^2}$ in each $L_2$ estimation algorithm to achieve $\PPr{\calE_1}=1-\frac{1}{m^2}$. 
Similarly, by setting the probability of failure $\delta=\frac{1}{m^2}$ in \thmref{thm:countsketch}, it follows that $\countsketch_a$ succeeds with probability $1-\frac{1}{m^2}$ and thus $\PPr{\calE_2}=1-\frac{1}{m^2}$. 
Finally, note that for window sizes of length $W$ with $W=\O{\frac{\log^5 m}{\alpha^2\eps^2}}$, we can simply store the entire set of active items. 
Thus we assume $W=\Omega\left(\frac{\log^5 m}{\alpha^2\eps^2}\right)$ so that $L_2(f)=\Omega\left(\frac{\log^2 m}{\alpha\eps}\right)$. 
Let $\PPr{\widehat{L_2}(f) \leq 2 L_2(f)}\leq 1/m^2$, then since $X\sim\Lap\left(\frac{1}{40\log m}\,\widehat{L_2}(f)\right)$ and $Y_k,Z_k\sim\Lap\left(\frac{\alpha}{75\log m}\,\widehat{L_2}(f)\right)$ for each $k\in[n]$, it follows from \claimref{claim:lap-HH} that $\PPr{X\le\frac{L_2(f)}{10}\,\wedge\,\max_{j\in[n]}(Y_j,Z_j)\le\frac{\alpha L_2(f)}{10}}\ge 1-\left(\frac{2}{m^2}+\frac{2}{m^{\frac{11}{4}}}\right)$. 
Hence, by a union bound, we have $\PPr{\calE}\ge 1-\frac{4}{m^2}-\frac{2}{m^{\frac{11}{4}}}$.  
\end{proof}

Before analyzing the privacy guarantees of \algref{alg:dp:sw:hh}, we must analyze the local sensitivity of its subroutines. 
\cite{BlockiGMZ22} first showed the following $\beta$-smooth upper bound on the local sensitivity of the frequency moment. 
For completeness, we include the proof below. 
\begin{lemma}[Smooth sensitivity of the frequency moment {\cite{BlockiGMZ22}}]
\lemlab{lem:smooth:sens:ftwo}
Let $\frakS$ be a data stream of length $m$ that induces a frequency vector $f$ and let $\widehat{L_2}(f)$ be the estimate of $L_2(f)$ output by the smooth histogram. 
Define the function $g(f)$ by 
\[g(f)=\begin{cases}\widehat{L_2}(f),&\text{if }\left(1-\frac{\eps}{500\log m}\right)L_2(f)\le\widehat{L_2}(f)\le\left(1+\frac{\eps}{500\log m}\right)L_2(f),\\
\left(1-\frac{\eps}{500\log m}\right)L_2(f),&\text{if }\widehat{L_2}(f)<\left(1-\frac{\eps}{500\log m}\right)L_2(f),\text{ and}\\
\left(1+\frac{\eps}{500\log m}\right)L_2(f),&\text{if }\widehat{L_2}(f)>\left(1+\frac{\eps}{500\log m}\right)L_2(f).
\end{cases}\]
Then the function $S(f)=\frac{\eps}{200\log m}\,g(f)+2$ is a $\beta$-smooth upper bound on the local sensitivity of $g(f)$ for $\beta\ge\frac{\eps}{150\log m}$, $\eps>\frac{1000\log m}{\sqrt{W}}$, and sufficiently large $W$.
\end{lemma}
\begin{proof}
Let $S(f)=\frac{\eps}{200\log m}\,g(f)+2$ and suppose $\|f-f'\|_1=1$. 
Moreover, suppose without loss of generality that $g(f)\ge g(f')$. 
We first show that Condition~\ref{it:smooth-1} of \defref{def:smooth-ub} holds.
The local sensitivity of $g(f)$ is at most 
\[|g(f)-g(f')|\le\left|\left(1+\frac{\eps}{500\log m}\right)L_2(f)-\left(1-\frac{\eps}{500\log m}\right)L_2(f')\right|.\]
By \lemref{lem:sens:ltwo}, we have $|L_2(f)-L_2(f')|\le 2$. 
Thus for $\eps>\frac{1000\log m}{\sqrt{W}}$ and sufficiently large $W$,
\begin{align*}
|g(f)-g(f')|&\le 2+\left(\frac{\eps}{500\log m}\right)L_2(f)+\left(\frac{\eps}{500\log m}\right)L_2(f')\\
&\le 2+\left(\frac{\eps}{500\log m}\right)(2L_2(f)+2)\\
&\le 2+\left(\frac{\eps}{225\log m}\right)\,L_2(f)\\
&\le 2+\left(\frac{\eps}{225\log m}\right)\left(\frac{1}{1-\eps/(500\log m)}\right)g(f)\\
&\le 2+\frac{\eps}{200\log m}\,g(f)=S(f).
\end{align*}
Thus Condition~\ref{it:smooth-1} of \defref{def:smooth-ub} holds.
We next show that Condition~\ref{it:smooth-2} of \defref{def:smooth-ub} holds.
Furthermore, from the above, we have 
\begin{align*}
S(f')&=\frac{\eps}{200\log m}\,g(f')+2\\
&\le\left(\frac{\eps}{200\log m}\right)\left(g(f)+2+\left(\frac{\eps}{200\log m}\right)\,L_2(f)\right)+2\\
&\le\left(1+\frac{\eps}{150\log m}\right)\left(\frac{\eps}{200\log m}\,g(f)+2\right)\\
&\le e^\beta\,S(f).
\end{align*}
Therefore, both conditions of \defref{def:smooth-ub} hold and it follows that the function $S(f)=\frac{\eps}{200\log m}\,g(f)+2$ is a $\beta$-smooth upper bound on the local sensitivity of $g(f)$ for $\beta\ge\frac{\eps}{150\log m}$, $\eps>\frac{1000\log m}{\sqrt{W}}$, and sufficiently large $W$.
\end{proof}
We next show a $\beta$-smooth upper bound on the local sensitivity for each estimated frequency output by \algref{alg:dp:sw:hh}. 
\begin{lemma}[Smooth sensitivity of the estimated frequency]
\lemlab{lem:smooth:sens:counter}
Let $\frakS$ be a data stream of length $m$ that induces a frequency vector $f$ and let $\widehat{f_k}$ be the estimate of the frequency of a coordinate $k\in[n]$ output by the smooth histogram. 
Define the function $h(f)$ by 
\[h(f)=\begin{cases}\widehat{f_k},&\text{if }f_k-\frac{\alpha^3\eps}{1000\log m} L_2(f)\le\widehat{f_k}\le f_k+\frac{\alpha^3\eps}{1000\log m} L_2(f),\\
f_k-\frac{\alpha^3\eps}{1000\log m} L_2(f),&\text{if }\widehat{f_k}<f_k-\frac{\alpha^3\eps}{1000\log m} L_2(f),\text{ and}\\
f_k+\frac{\alpha^3\eps}{1000\log m} L_2(f),&\text{if }\widehat{f_k}>f_k+\frac{\alpha^3\eps}{1000\log m} L_2(f).
\end{cases}\]
Then the function $S(f)=\frac{\alpha^3\eps}{200\log m}\,h(f)+2$ is a $\beta$-smooth upper bound on the local sensitivity of $h(f)$ for $\beta\ge\frac{\alpha^3\eps}{150\log m}$, $\eps>\frac{1000\log m}{\sqrt{W}\alpha^3}$, and sufficiently large $W$.
\end{lemma}
\begin{proof}
The proof is almost identical to that of \lemref{lem:smooth:sens:ftwo}; we include it for completeness. 
Let $S(f)=\frac{\alpha^3\eps}{200\log m}\,h(f)+2$ and suppose $\|f-f'\|_1=1$. 
Moreover, suppose without loss of generality that $h(f)\ge h(f')$. 
We first show that Condition~\ref{it:smooth-1} of \defref{def:smooth-ub} holds.
The local sensitivity of $h(f)$ is at most 
\[|h(f)-h(f')|\le|f_k-f'_k|+\frac{\alpha^3\eps}{1000\log m}\,(L_2(f)+L_2(f')).\]
By \lemref{lem:sens:ltwo}, we have $|L_2(f)-L_2(f')|\le 2$. 
Since $|f_k-f'_k|\le 2$, we further have for $\eps>\frac{1000\log m}{\sqrt{W}\alpha^3}$ and sufficiently large $W$, 
\begin{align*}
|h(f)-h(f')|&\le 2+\frac{\alpha^3\eps}{1000\log m}\,(L_2(f)+L_2(f)+2)\\
&\le 2+\left(\frac{\alpha^3\eps}{1000\log m}\right)(2L_2(f)+2)\\
&\le 2+\left(\frac{\alpha^3\eps}{400\log m}\right)\,L_2(f)\\
&\le 2+\left(\frac{\alpha^3\eps}{400\log m}\right)\left(\frac{1}{1-(\alpha^3\eps)/(1000\log m)}\right)h(f)\\
&\le 2+\frac{\alpha^3\eps}{200\log m}\,h(f)=S(f).
\end{align*}
Thus Condition~\ref{it:smooth-1} of \defref{def:smooth-ub} holds.
We next show that Condition~\ref{it:smooth-2} of \defref{def:smooth-ub} holds.
Furthermore, from the above, we have 
\begin{align*}
S(f')&=\frac{\alpha^3\eps}{200\log m}\,h(f')+2\\
&\le\left(\frac{\alpha^3\eps}{200\log m}\right)\left(h(f)+2+\left(\frac{\alpha^3\eps}{400\log m}\right)\,L_2(f)\right)+2\\
&\le\left(1+\frac{\alpha^3\eps}{150\log m}\right)\left(\frac{\alpha^3\eps}{200\log m}\,h(f)+2\right)\\
&\le e^\beta\,S(f).
\end{align*}
Therefore, both conditions of \defref{def:smooth-ub} hold and it follows that the function $S(f)=\frac{\alpha^3\eps}{200\log m}\,h(f)+2$ is a $\beta$-smooth upper bound on the local sensitivity of $h(f)$ for $\beta\ge\frac{\alpha^3\eps}{150\log m}$, $\eps>\frac{1000\log m}{\sqrt{W}\alpha^3}$, and sufficiently large $W$.
\end{proof}

With the structural results on smooth sensitivity in place, we now show that \algref{alg:dp:sw:hh} is $(\eps,\delta)$-differentially private. 
\begin{lemma}
\lemlab{lem:dp:sw:hh}
There exists an algorithm (see \algref{alg:dp:sw:hh}) that is $(\eps,\delta)$-differentially private for $\alpha\in(0,1)$, $\eps>\frac{1000\log m}{\sqrt{W}\alpha^3}$, and $\delta>\frac{6}{m^2}$.
\end{lemma}
\begin{proof}
Let $\calE$ be the event that (1) the smooth histogram data structure does not fail on either stream, (2) all instances of $\countsketch$ do not fail, and (3)$X\le\frac{L_2(f)}{10}$ and $\max_{j\in[n]}(Y_j,Z_j)\le\frac{\alpha L_2(f)}{10}$. 
By \lemref{lem:fail:prob}, we have $\PPr{\calE}\ge 1-\frac{4}{m^2}-\frac{2}{m^{{11}/{4}}}$. 
Conditioned on $\calE$, then by \lemref{lem:hh:high}, the algorithm releases at most $\frac{10}{\alpha^2}$ frequencies. 

Let $f$ be the frequency vector defined by the active elements of the sliding window. 
Let $g(f)$ be the function defined in \lemref{lem:smooth:sens:ftwo}. 
Then by \lemref{lem:smooth:sens:ftwo}, we have that $\frac{\eps}{125\log m}\,g(f)$ is a $\beta$-smooth upper bound on the local sensitivity of $g(f)$, for $\beta=\frac{\eps}{2\ln(2m^3)}$ and sufficiently large $m$. 
Thus by \thmref{thm:smooth:laplace} it suffices to add Laplacian noise $X\sim\Lap\left(\frac{1}{40}\,g(f)\right)$ to the estimated norm $\widehat{L_2}$ output by the smooth histogram to achieve $\left(\frac{\eps}{3},\frac{\delta_1}{3}\right)$-differential privacy where $\frac{\delta_1}{3}=\frac{1}{m^3}$.

Similarly by \lemref{lem:smooth:sens:counter}, we have that $\frac{\alpha^3\eps}{200\log m}\,h(f)$ is a $\beta$-smooth upper bound on the local sensitivity of each estimated frequency for $\beta=\frac{\eps}{2\ln(2m^3)}$ and sufficiently large $m$. 
Thus by \thmref{thm:smooth:laplace} it suffices to add Laplacian noise $Y_k\sim\Lap\left(\frac{\alpha}{75\log m}\,h(f)\right)$ to each estimated frequency $\widehat{f_k}$ output by the smooth histogram to achieve $\left(\frac{\eps \alpha^2}{30},\frac{\delta_2}{3}\right)$-differential privacy where $\frac{\delta_2}{3}=\frac{1}{m^3}$.
Since we only release at most $\frac{10}{\alpha^2}$ estimated frequencies, then by \thmref{thm:dp:laplace} and \thmref{thm:dp:comp}, the estimated frequencies for each heavy-hitter are $(\frac{\eps}{3},\frac{\delta_2}{3})$-differentially private. 
Finally, we only release the estimated frequencies that are above a threshold $\frac{3\alpha}{4}\,(\widehat{L_2}+X)+Y_k$, which corresponds to a limited histogram query. 
Hence by \thmref{thm:dp:comp}, the algorithm would be $(\eps,\frac{2\delta}{3})$-differentially private when using the functions $g(f)$, to which we do not have access. 
Instead, we note that conditioned on $\calE$, we have that $g(f)=\widehat{L_2}(f)$ with probability $1-\frac{4}{m^2}-\frac{2}{m^{\frac{11}{4}}}$. Setting $\frac{\delta_3}{3}=\frac{5}{m^2}>\frac{4}{m^2}-\frac{2}{m^{\frac{11}{4}}}$, and consequently, setting $\delta=\frac{6}{m^2}> \frac{\delta_1}{3}+\frac{\delta_2}{3}+\frac{\delta_3}{3} = \frac{2}{m^3}+\frac{5}{m^2}$, this is absorbed into the failure probability. 
Thus the entire algorithm is $(\eps,\delta)$-differentially private. 
\end{proof}

Since \algref{alg:dp:sw:hh} further adds Laplacian noise $Z\sim\Lap\left(\frac{\alpha}{75\log m}\,\widehat{L_2}(f)\right)$ to each $\widehat{f_k}$ with $k\in\calL$, then \lemref{lem:hh:high} implies that the additive error to $f_k$ is $\frac{\alpha}{50\log m}\,L_2(f)+\Lap\left(\frac{\alpha}{75\log m}\,\widehat{L_2}(f)\right)$ for each reported coordinate $k\in[n]$. 

Thus we now have our main result for differentially private $L_2$-heavy hitters in the sliding window model. 

\thmswmain*
\begin{proof}
Let $\calE$ be the event that (1) the smooth histogram data structure does not fail on either stream, (2) all instances of $\countsketch$ do not fail, and (3) $\max_{j\in[n]}(X,Y_j,Z_j)\le\frac{\alpha L_2(f)}{10}$. 
By \lemref{lem:fail:prob}, we have $\PPr{\calE}\ge 1-\frac{1}{m^c}$. 
Conditioned on $\calE$, the first and second claims follow immediately from \lemref{lem:hh:low} and \lemref{lem:hh:high}. 
The third claim follows from \lemref{lem:dp:sw:hh}. 

To analyze the space complexity, note that by \thmref{thm:ams}, each $\left(1+\frac{\eps}{500\log m}\right)$-approximation algorithm for $L_2$ with failure probability $\delta=\frac{1}{m^c}$ uses $\O{\frac{1}{\eps^2}\log^2 m}$ bits of space. 
Similarly by \thmref{thm:countsketch}, each instance of $\countsketch$ with threshold $\frac{\alpha^3\eps}{500\log m}$ uses $\O{\frac{1}{\alpha^6\eps^2}\log^4 m}$ bits of space. 
By the smooth histogram data structure, the timestamps $t_1,\ldots,t_s$ satisfy $L_2(t_i,m)\ge\left(1+\frac{\eps}{500\log m}\right)\,L_2(t_{i+2},m)$. 
By the smoothness of $L_2$ in \lemref{lem:fp:smooth}, this requires the invariant that $L_2(t_i,t)\le\left(1+\left(\frac{\eps}{1000\log m}\right)^2\right)\,L_2(t_{i+1},t)$ for all $t\in[m]$. 
Since the $L_2$ of the entire stream is at most $m$, $s=\O{\frac{1}{\eps^2}\log^3 m}$. 
Thus there are at most $\O{\frac{1}{\eps^2}\log^3 m}$ instances of $\countsketch$ and $\ams$. 
Hence the total space used by $\countsketch$ and $\ams$ is $\O{\frac{1}{\alpha^6\eps^4}\log^7 m}$. 
By \lemref{lem:counter}, the total space used by the frequency estimation algorithms $\counter$ across all $\O{\frac{1}{\eps}\log m}$ instances is also at most $\O{\frac{1}{\alpha^6\eps^4}\log^7 m}$. 
Therefore, the total space used by the algorithm is $\O{\frac{1}{\alpha^6\eps^4}\log^7 m}$ bits of space. 

To analyze the time complexity, note that by \thmref{thm:countsketch} and \lemref{lem:counter}, $\O{\frac{\log n}{\alpha^2\eps^2}+\log m}$ operations are required per update, per instance of $\countsketch$ and $\counter$. 
The update time then follows from the fact that there are $\O{\frac{1}{\eps^2}\log^3 m}$ simultaneous instances. 
\end{proof}

Finally, observe that in the proof of \lemref{lem:dp:sw:hh}, each instance of the frequency estimation algorithm is required to be $\left(\frac{\eps}{30\alpha^2},\frac{\delta}{3}\right)$-differentially private, since we release $\O{\frac{1}{\alpha^2}}$ frequencies, by the standard differential privacy composition theorem, i.e., \thmref{thm:dp:comp}. 
We can instead use advanced composition, i.e., \thmref{thm:dp:adv:comp}, to require each instance of the frequency estimation algorithm to be $\left(\O{\frac{\eps}{\alpha\sqrt{\log m}}},\frac{\delta}{3}\right)$-differentially private, for $\eps=\O{\alpha^2}$. 
In this case, we can require each frequency estimation algorithm to have a cruder approximation guarantee, e.g., using $\countsketch$ with threshold $\O{\frac{\alpha^2\eps}{\log^{3/2} m}}$ rather than threshold $\frac{\alpha^3\eps}{500\log m}$. 

\begin{restatable}{theorem}{thmswalg}
\thmlab{thm:sw:alg}
For any $\alpha, c>0$, window parameter $W$ on a stream of length $m$ that induces a frequency vector $f\in\mathbb{R}^n$ in the sliding window model and privacy parameter  $\eps$, there exists an algorithm (see \algref{alg:dp:sw:hh}) such that:
\begin{enumerate}
\item
With probability at least $1-\frac{1}{m^c}$, the algorithm outputs a list $\calL$ such that $k\in\calL$ for each $k\in[n]$ with $f_k\ge\alpha\,L_2(f)$ and $j\notin\calL$ for each $j\in[n]$ with $f_j\le\frac{\alpha}{2}\,L_2(f)$. 
\item
With probability at least $1-\frac{1}{\poly(m)}$, we simultaneously have $|f_k-\widetilde{f_k}|\le\frac{\alpha}{4}\,L_2(f)$ for all $k\in\calL$, where $\widetilde{f_k}$ denotes the noisy approximation of $f_k$ output by \algref{alg:dp:sw:hh}.
\item
The algorithm is $(\eps,\delta)$-differentially private for $\delta=\frac{1}{m^c}$.
\item
The algorithm uses $\O{\frac{\log^8 m}{\alpha^4\eta^3}}$ bits of space and $\O{\frac{\log^{9/2} m}{\alpha^2\eta^3}}$ operations per update where $\eta = \max\{1, \eps\}$. 
\end{enumerate}
\end{restatable}

\section{Better \texorpdfstring{$L_1$}{L1} Heavy-Hitter Algorithm}
\seclab{sec:lone}
In this section, we show how our main technical ideas can be simplified and applied to give a better $L_1$-heavy hitter algorithm than the one-shot sliding window algorithm given in \cite{Upadhyay19}. 
Due to the linearity of $L_1$, our algorithm for differentially private $L_1$-heavy hitters in the sliding window model is significantly simpler than the $L_2$-heavy hitters algorithm. 
For starters, each set of $c>0$ updates must contribute exactly $c$ to the $L_1$ norm, whereas their contribution to the $L_2$ norm depends on the particular coordinates they update. 
Therefore, not only do we not require an algorithm to approximate the $L_1$ norm of the active elements of the sliding window, but also we can fix a set of static timestamps in the smooth histogram, so we do not need to perform the same analysis to circumvent the sensitivity of the timestamps. 
Instead, it suffices to initialize a \emph{deterministic} $L_1$-heavy hitter algorithm at each timestamp and maintain deterministic counters for each reported heavy-hitter. 
Pure differential privacy then follows from the lack of failure conditions in the subroutines, which was not possible for $L_2$-heavy hitters. 

We first require the following \emph{deterministic} algorithm for $L_1$-heavy hitters. 
\begin{theorem}[Heavy-hitter algorithm MisraGries {\cite{MisraG82}}]
\thmlab{thm:misragries}
Given an accuracy parameter $\alpha>0$, there exists a one-pass algorithm $\misragries$ for the $L_1$-heavy hitter problem, using $\O{\frac{1}{\alpha}\log m}$ bits of space and update time on a stream of length $m$ and a universe of size $n$.
\end{theorem}
The global $L_1$-sensitivity of the $\misragries$ algorithm is upper bounded as follows:
\begin{fact}
\factlab{fact:mg:sens}
$\misragries$ on a stream of length $m$ and threshold $\alpha$ has $L_1$ sensitivity $\alpha m$. 
\end{fact}

The algorithm for differentially private $L_1$-heavy hitters in the sliding window model is much simpler than the $L_2$-heavy hitters algorithm due to the linearity of $L_1$. 
For instance, because we know that each set of $c$ updates contributes exactly $c$ to the $L_1$ norm, then it suffices to maintain a static set of timestamps in the smooth histogram. 
Therefore, the timestamps do not change across neighboring datasets, so we do not need to analyze the sensitivity of the timestamps. 
Similarly, we do not need an algorithm to approximate the $L_1$ norm of the substream starting at each timestamp, since it is also fixed. 
Hence, we do not need analysis for the smooth sensitivity of an $L_1$ norm approximation algorithm and the subsequent statistical distance analysis that we needed for $\ams$ for $L_2$ norm approximation. 
Instead, it suffices to maintain a counter for each reported heavy-hitter by some instance of $\misragries$ starting at each of the timestamps. 
We then output the coordinates whose estimated frequencies by the counters surpass a fixed threshold in terms of $W$, which again is deterministic due to the linearity of $L_1$. 
We give the algorithm in full in \algref{alg:dp:sw:hh:one}. 

\begin{algorithm}[!htb]
\caption{Differentially private sliding window algorithm for $L_1$-heavy hitters}
\alglab{alg:dp:sw:hh:one}
\begin{algorithmic}[1]
\Require{Stream $\frakS$, threshold/accuracy parameter $\alpha\in(0,1)$, differential privacy parameters $\eps,\delta>0$, window parameter $W>0$, size $n$ of the underlying universe, upper bound $m$ on stream length}
\Ensure{A list $\calL$ of $L_1$ heavy-hitters with approximate frequencies}
\State{\textbf{Process} the stream $\frakS$, maintaining timestamps $t_1,\ldots,t_s$ at each time $t\in[m]$ so that for each $i\in[s]$, either $i=1$, $i=s$, or $t-t_{i-1}+1>\left(1+\frac{1}{100}\right)(t-t_{i+1}+1)$ through a smooth histogram}
\State{\textbf{Implement} heavy-hitter algorithm $\misragries$ on the substream starting at $t_i$ for each $i\in[s]$ with threshold $\frac{\alpha}{16}$}
\State{\textbf{Set} $a=\max\{i\in[s]:t_i\le m-W+1\}$ on window query $W>0$}
\For{each heavy-hitter $k\in[n]$ reported by $\misragries$ starting at $t_a$}
\State{\textbf{Run} $\counter$ with additive error $\frac{\alpha}{32}\,(t-t_a+1)$}
\State{Let $\widehat{f_k}$ be the approximate frequency reported by $\counter$}
\State{$Z_k\gets\Lap\left(\frac{1}{\eps\alpha\log m}\right)$, $\widetilde{f_k}=\widehat{f_k}+Z_k$}
\If{$\widetilde{f_k}\ge\frac{3\alpha}{4}\,W$}
\State{$\calL\gets\calL\cup\{(k,\widetilde{f_k})\}$}
\EndIf
\EndFor
\State{\Return $\calL$}
\end{algorithmic}
\end{algorithm}

We first analyze the accuracy of the estimated frequency for each item output by \algref{alg:dp:sw:hh:one}. 
\begin{lemma}[Accuracy of frequency estimation]
\lemlab{lem:correct:sw:hh:one}
Let $\calE$ be the event that $\max_{j\in[n]}(Z_j)\le\frac{\alpha}{16}\,W$. 
Then for each $k\in[n]$ with $f_k\ge\alpha\,L_1(f)$, we have $k\in\calL$. 
Moreover, for each item $k\in\calL$, 
\[|f_k-\widetilde{f_k}|\le\frac{\alpha}{8}\,W.\]
\end{lemma}
\begin{proof}
Consider a time $t\in[m]$ that induces a frequency vector $f$. 
We have $L_1(f)\le L_1(t_a:t)\le\left(1+\frac{1}{100}\right)L_1(f)$. 
Hence for each $k\in[n]$ with $f_k\ge\alpha\,L_1(f)$, we have $f_k\ge\frac{\alpha}{16}\,L_1(t_a:t)$ and thus $k$ will be reported by the instance of $\misragries$ starting at time $t_a$. 
For each item $k$ reported by $\misragries$ starting at time $t_a$, $\counter$ reports an estimated frequency $\widehat{f_k}$ such that 
\[|f_k-\widehat{f_k}|\le\frac{\alpha}{32}\,L_1(t_a:t)\le\frac{\alpha}{16}\,W.\]
Conditioned on the event that $\max_{j\in[n]}(Z_j)\le\frac{\alpha}{16}\,W$, then we have $|\widetilde{f_k}-\widehat{f_k}|\le\frac{\alpha}{16}\,W$. 
Hence by triangle inequality,
\[|f_k-\widetilde{f_k}|\le\frac{\alpha}{16}\,L_1(t_a:t)\le\frac{\alpha}{8}\,W\]
for each item $k\in\calL$.
\end{proof}

We now show that \algref{alg:dp:sw:hh:one} guarantees pure differential privacy. 
By comparison, the algorithm of \cite{Upadhyay19} can only guarantee $(\eps,\delta)$-differential privacy due to the failure probability of their randomized subroutines. 
\begin{lemma}
\lemlab{lem:dp:sw:hh:one}
\algref{alg:dp:sw:hh:one} is $(\eps,0)$-differentially private for any constant $\eps\in(0,1]$.
\end{lemma}
\begin{proof}
Note that unlike for the $L_2$-heavy hitter algorithm where the times $t_1,\ldots,t_s$ are determined by the output of the randomized algorithms for $L_2$-norm estimation, the times $t_1,\ldots,t_s$ are deterministic since the $L_1$ norm of the frequency vector defined from a time $t$ to a time $t'$ is simply $t'-t+1$. 
Thus the sensitivity of the estimated frequency of each element is the sensitivity of the instance of $\misragries$ that starts at time $t_a$. 
Consequently by \factref{fact:mg:sens}, the sensitivity of each estimated frequency is at most $2$. 
Hence by \thmref{thm:dp:laplace}, it suffices to add Laplacian noise that scales with $\frac{2}{\eps}$ to each estimated frequency to preserve $(\eps,0)$-differential privacy. 
\end{proof}
We now give the full guarantees of \algref{alg:dp:sw:hh:one}.

\begin{restatable}{theorem}{thmswhhone}
\thmlab{thm:sw:hh:one}
For any $\alpha, c>0$, window parameter $W$ on a stream of length $m$ that induces a frequency vector $f\in\mathbb{R}^n$ in the sliding window model and constant privacy parameter $\eps$, there exists an algorithm (see \algref{alg:dp:sw:hh:one}) such that:
\begin{enumerate}
\item
With probability at least $1-\frac{1}{m^c}$, the algorithm outputs a list $\calL$ such that $k\in\calL$ for each $k\in[n]$ with $f_k\ge\alpha\,W$ and $j\notin\calL$ for each $j\in[n]$ with $f_j\le\frac{\alpha}{2}\,W$. 
\item
With probability at least $1-\frac{1}{m^c}$, we simultaneously have $|f_k-\widetilde{f_k}|\le\frac{\alpha}{4}\,W$ for all $k\in\calL$ where $\widetilde{f_k}$ denotes the noisy approximation of $f_k$ output by \algref{alg:dp:sw:hh:one}. 
\item
The algorithm is $(\eps,0)$-differentially private.
\item
The algorithm uses $\O{\frac{\log^2 m}{\alpha}}$ bits of space and $\O{\frac{\log^2 m}{\alpha}}$ update time per arriving symbol. 
\end{enumerate}
\end{restatable}
\begin{proof}
The first three statements follow from \lemref{lem:correct:sw:hh:one} and \lemref{lem:dp:sw:hh:one}. 
Thus it remains to analyze the space complexity. 
Note that for each time $t_i$ with $i\in[s]$, we have either $i=1$, $i=s$, or $t-t_{i-1}+1>\left(1+\frac{1}{100}\right)(t-t_{i+1}+1)$, so that $s=\O{\log m}=\O{\log n}$. 
Each time $t_i$ corresponds to a separate instance of $\misragries$ that uses $\O{\frac{1}{\alpha}\log n}$ bits of space and update time per arriving symbol. 
The instances of $\counter$ associated with each time $t_i$ can also maintain $\O{\frac{1}{\alpha}}$ times, which use $\O{\frac{1}{\alpha}\log m}$ space and update time for each time $t_i$.  
Hence, the total space and update time is $\O{\frac{\log^2 m}{\alpha}}$. 
\end{proof}

\section*{Acknowledgements}
We would like to thank Sofya Raskhodnikova for clarifying discussions about smooth sensitivity.

\def\shortbib{0}
\bibliographystyle{alpha}
\bibliography{references}


\appendix

\section{Continual Release of \texorpdfstring{$L_1$}{L1} Heavy-Hitters}
\applab{app:continual:lone}
Our algorithm consists of $L:=\O{\log W}=\O{\log n}$ levels of subroutines. 
In each level $\ell\in[L]$, we split the stream into continuous blocks of length $S_\ell:=2^{\ell-2}\cdot\frac{\alpha\sqrt{W}}{100\log W}$. 
Given a threshold parameter $\alpha>0$, for each block in level $\ell$, we run an instance of $\misragries$ with threshold $\frac{1}{2^{\ell+1}L}$. 
At the end of the stream, we stitch together a sketch of the underlying dataset represented by the sliding window through a binary tree mechanism. 

\begin{algorithm}[!htb]
\caption{Differentially private sliding window algorithm for heavy-hitters}
\alglab{alg:dp:continuous}
\begin{algorithmic}[1]
\Require{Stream $\frakS$ of length $m$, accuracy parameter $\alpha>0$, differential privacy parameter $\eps>0$, window parameter $W>0$}
\Ensure{Continuous release of list $\calL$ of $L_1$ heavy-hitters with approximate frequencies}
\State{$L\gets\ceil{\log\frac{100\eps}{\alpha}\,\sqrt{W}+\log\log W}+2$, $\varphi\gets\min(\eps,1)$}
\For{$\ell\in[L]$}
\State{\textbf{Partition} the stream into blocks of length $S_\ell:=2^{\ell-2}\cdot\frac{\alpha\sqrt{W}}{100\log W}$}
\State{\textbf{Run} $\misragries$ on each block of length $S_\ell$ with threshold $\frac{\varphi\alpha\sqrt{W}}{16S_\ell\cdot\log^3 W}$}
\EndFor
\For{each time $i\in[m]$}
\State{\textbf{Update} each $\misragries$ algorithm}
\State{\textbf{Build} a binary tree mechanism and output the reported heavy-hitters}
\For{each active block $B_1,\ldots,B_L$ in the binary tree mechanism}
\State{\textbf{Set} $\widehat{g_i^{(r)}}$ be the estimate of each $i\in[n]$ in $B_r$}
\State{$\widehat{f_i^{(r)}}\gets\widehat{g_i^{(r)}}+\Lap\left(\frac{\alpha\sqrt{W}}{16\log^2 W}\right)$}
\EndFor
\State{\Return $\widehat{f_i}=\sum_{r=1}^L\widehat{f_i^{(r)}}$}
\EndFor
\end{algorithmic}
\end{algorithm}

From \thmref{thm:misragries}, we have the following accuracy guarantees:
\begin{corollary}
\corlab{cor:error:lvl}
Suppose $\misragries$ at level $\ell\in[L]$ on a substream induced by the updates between $\left[t+1,t+2^{\ell-2}\cdot\frac{\sqrt{W}}{100\log W}\right]$ outputs an estimate $\widehat{g_i}$ for the frequency $g_i$ of item $i\in[n]$ between $\left[t+1,t+2^{\ell-2}\cdot\frac{\sqrt{W}}{100\log W}\right]$. 
Then $|g_i-\widehat{g_i}|\le\frac{\alpha\sqrt{W}}{16\log^3 W}$. 
\end{corollary}

\begin{lemma}[Accuracy of frequencies]
\lemlab{lem:cont:acc}
For all $i\in[n]$, we have that $|f_i-\widehat{f_i}|\le\frac{\alpha\sqrt{W}}{8}$.
\end{lemma}
\begin{proof}
Let $t$ be the smallest multiple of $\frac{\alpha\sqrt{W}}{200\log W}$ with $t\ge m-W+1$. 
In other words, $t$ is the first starting location of a block at level $1$ among the active items. 
For each $i\in[n]$, let $\tilde{f_i}$ be the frequency of $i$ from time $t$ to $m$. 
Since $(m-W+1)-t<\frac{\alpha\sqrt{W}}{200\log W}$, then we have 
\[f_i-\tilde{f_i}<\frac{\alpha\sqrt{W}}{200\log W}\] 
for all $i\in[n]$. 

Our binary tree mechanism partitions the stream from time $t$ to $m$ into at most two disjoint blocks in each level. 
Suppose these blocks are $B_1,\ldots,B_L$ across the $L$ levels.  
For each $i\in[n]$ and $r\in[L]$, let $g_i^{(r)}$ be the frequency of $i$ within the substream allocated to $B_r$. 
Similarly, let $\widehat{g_i^{(r)}}$ be the estimated frequency of $i$ within the substream by the corresponding instance of $\misragries$. 
By \corref{cor:error:lvl}, we have that $|g_i^{(r)}-\widehat{g_i^{(r)}}|<\frac{\alpha\sqrt{W}}{16\log^3 W}$. 
Thus we have by triangle inequality,
\begin{align*}
\left|\tilde{f_i}-\sum_{r=1}^L\widehat{g_i^{(r)}}\right|&=\left|\sum_{r=1}^L g_i^{(r)}-\sum_{r=1}^L\widehat{g_i^{(r)}}\right|\\
&\le\sum_{r=1}^L\left|g_i^{(r)}-\widehat{g_i^{(r)}}\right|\le\sum_{r=1}^L\frac{\alpha\sqrt{W}}{16\log W}\\
&\le\frac{L\alpha\sqrt{W}}{16\log W}\le\frac{\alpha\sqrt{W}}{16}
\end{align*}
Hence by another triangle inequality,
\[\left|f_i-\sum_{r=1}^L\widehat{g_i^{(r)}}\right|\le\left|\tilde{f_i}-\sum_{r=1}^L\widehat{g_i^{(r)}}\right|+|f_i-\tilde{f_i}|\le\frac{\alpha\sqrt{W}}{8}.\]
The claim then follows from setting $\widehat{f_i}=\sum_{r=1}^L\widehat{g_i^{(r)}}$. 
\end{proof}

\begin{theorem}
\thmlab{thm:sw:continuous}
Given threshold/accuracy parameter $\alpha>0$, privacy parameter $\eps=\O{1}$, and window parameter $W$ on a stream of length $m$ that induces a frequency vector $f \in \mathbb{R}^n$ in the sliding window model, there exists an algorithm such that:
\begin{enumerate}
\item
The algorithm continually outputs a list $\calL$ such that $k\in\calL$ for each $k\in[n]$ with $f_k\ge\alpha\sqrt{W}$ and with high probability, we have $|f_k-\widehat{f_k}|\le\frac{\alpha\sqrt{W}}{2}$ for each $k\in\calL$.  
\item
The algorithm is $(\eps,0)$-differentially private.
\item
The algorithm uses $\O{\frac{\sqrt{W}\log^4 W}{\eps\alpha}}$ bits of space and operations per update. 
\end{enumerate}
\end{theorem}
\begin{proof}
Consider \algref{alg:dp:continuous}. 
For each $\ell\in[L]$, the stream is partitioned into blocks of length $S_\ell:=2^{\ell-2}\cdot\frac{\alpha\sqrt{W}}{100\log W}$ and then $\misragries$ is run on each block of length $S_\ell$ with threshold $\frac{\varphi\alpha\sqrt{W}}{16\cdot S_\ell\log^3 W}$, with $\varphi=\min(\eps,1)$. 
By \factref{fact:mg:sens}, the sensitivity of each block is $\frac{\varphi\alpha\sqrt{W}}{16\log^3 W}$. 
Thus by \thmref{thm:dp:laplace}, it suffices to add Laplacian noise $\Lap\left(\frac{\sqrt{\alpha W}}{8\log^2 W}\right)$ to each block to achieve $\left(\frac{\eps}{2\log W},0\right)$-differential privacy. 
Then by composition, i.e., \thmref{thm:dp:comp}, we obtain $(\eps,0)$-differential privacy. 

Moreover, we have that the sum of $\O{\log W}$ random variables drawn from the distribution $\Lap\left(\frac{\alpha\sqrt{W}}{8\log^2 W}\right)$ is at most $\frac{\alpha\sqrt{W}}{4}$ with high probability. 
Thus by \lemref{lem:cont:acc}, we have $|f_k-\widehat{f_k}|\le\frac{\alpha\sqrt{W}}{2}$ for each $k\in\calL$ with high probability. 

By \thmref{thm:misragries}, running $\misragries$ on each block of size $S_\ell$ uses 
\[\O{\frac{S_\ell\cdot\log^3 W}{\eps\alpha\sqrt{W}}}=\O{\frac{2^\ell\cdot\log^2 W}{\eps}}\]
bits of space and update time per operation. 
There are at most $\frac{W}{S_\ell}=\O{\frac{\sqrt{W}\log W}{2^\ell\alpha}}$ blocks of size $S_\ell$, thus the total space and update time per operation across the blocks of size $S_\ell$ is $\O{\frac{\sqrt{W}\log^3 W}{\eps\alpha}}$. 
Finally, we have $\ell\in[L]$, where $L=\O{\log W}$, so the total space and update time per operation of \algref{alg:dp:continuous} is $\O{\frac{\sqrt{W}\log^4 W}{\eps\alpha}}$. 
\end{proof}

Finally, we remark that for a window of $W$ updates, we have $L_2(t-W+1:t)\ge\sqrt{W}$. 
Thus \thmref{thm:sw:continuous} outputs a list $\calL$ such that $k\in\calL$ for each $k\in[n]$ with $f_k\ge\alpha\sqrt{W}$ and in particular, $f_k\ge\alpha\,L_2(t-W+1:t)$. 
Moreover, the algorithm achieves additive error $\frac{\alpha\sqrt{W}}{2}\le\frac{\alpha\,L_2(t-W+1:t)}{2}$ for each $k\in\calL$.  
Therefore, \thmref{thm:sw:continuous} not only improves upon the continual $L_1$-heavy hitters of \cite{Upadhyay19}, but it also solves the continual $L_2$-heavy hitters problem. 
\end{document}